\documentclass{article}

\usepackage{fullpage}
\usepackage{hyperref}
\usepackage{enumitem}
\usepackage{authblk}

\usepackage{amsmath}
\usepackage{amssymb}

\usepackage{graphicx}
\usepackage{epstopdf}
\usepackage{caption,subcaption}

\usepackage{tikz,tikz-3dplot}
\usetikzlibrary{calc,shapes}
\usetikzlibrary{decorations.pathreplacing} 

\usepackage{pgfplots}
\pgfplotsset{compat=1.8}
\usepgfplotslibrary{fillbetween}

\tikzset{
  queueps/.style={
    rectangle split,
    rectangle split parts=4,
    draw,
    rectangle split horizontal=false,
    text height=0cm,
    text depth=0cm,
    inner ysep=0.0245cm,
    inner xsep=.5cm,
    rounded corners,
  }
}

\usepackage{xcolor}
\xdefinecolor{myred}{RGB}{235,62,62}
\xdefinecolor{myblue}{RGB}{92,92,235}
\xdefinecolor{myorange}{RGB}{245,138,42}

\usepackage{amsthm}

\newtheorem{theorem}{Theorem}
\newtheorem{lemma}{Lemma}
\newtheorem{coro}{Corollary}

\theoremstyle{definition}
\newtheorem{definition}{Definition}
\newtheorem{ass}{Assumption}

\theoremstyle{remark}
\newtheorem{example}{Example}
\newtheorem*{num}{Numerical application}

\begin{document}

\title{Poly-Symmetry in Processor-Sharing Systems\thanks{The final publication is available at Springer via \url{http://dx.doi.org/10.1007/s11134-017-9525-2}.}}

\author[1]{Thomas Bonald}
\author[1]{C\'eline Comte}
\author[2]{Virag Shah\thanks{T. Bonald, C. Comte and V. Shah are members of LINCS, see \url{http://www.lincs.fr}.}}
\author[3]{Gustavo {de Veciana}}

\affil[1]{T\'el\'ecom ParisTech, Universit\'e Paris-Saclay, France}
\affil[2]{Microsoft Research - Inria Joint Centre, France}
\affil[3]{Department of ECE, The University of Texas at Austing, USA}

\date{\today}

\maketitle

\begin{abstract}
  We consider a system of processor-sharing queues with state-dependent service rates. These are allocated according to balanced fairness within a polymatroid capacity set. Balanced fairness is known to be both insensitive and Pareto-efficient in such systems, which ensures that the performance metrics, when computable, will provide robust insights into the real performance of the system considered. We first show that these performance metrics can be evaluated with a complexity that is polynomial in the system size if the system is partitioned into a finite number of parts, so that queues are exchangeable within each part and asymmetric across different parts. This in turn allows us to derive stochastic bounds for a larger class of systems which satisfy less restrictive symmetry assumptions. These results are applied to practical examples of tree data networks, such as backhaul networks of Internet service providers, and computer clusters. \\
  {\bf Keywords:} Processor-sharing queueing systems, performance, balanced fairness, poly-symmetry.
\end{abstract}

\section{Introduction}
\label{intro}

Systems of processor-sharing queues with state-dependent service rates
have been extensively used
to model a large variety of real communication and computation systems like
content delivery systems \cite{SV15,SV16},
computer clusters \cite{BC17-1,G16}
and data networks \cite{VL01,MaR00}.
They are natural models for such 
real systems as they capture the complex interactions between different jobs 
and also have a promise of analytical tractability of user performance when subject to stochastic loads.
Indeed, in the past two decades researchers have 
been able to obtain explicit performance expressions and bounds for several such systems, see
\cite{BP04-1,BP04-2,BV04,BV05,JoV11,KMW09,MaR00,SV15,SV16}.

However, few performance results scale well with the system size.
Those that do rely on restrictive assumptions
related to the topology or the symmetry of the system
\cite{MaR00,SV16}.
One of the main goals of this paper is to provide scalable performance results
for a class of processor-sharing systems which find applications in bandwidth-sharing networks and computer clusters. 

One of the key features of processor-sharing systems is the allocation of the service rates per queue in each state.
A particular class of resource allocations
which is more tractable for performance analysis
is characterized by the balance property
which constrains the relative gain in the service rate at one queue
when we remove a job from another queue.
Processor-sharing systems where the resource allocation satisfies this property
are called Whittle networks \cite{S99}.
In particular, if the service rates are constrained by some capacity set,
corresponding to the resources of the real system considered,
then there exists a unique policy which satisfies the balance property while being efficient,
namely balanced fairness \cite{BP03-1}.
In this paper we focus on systems which are constrained by a polymatroid capacity set \cite{F05,SV15}
and operate under balanced fair resource allocation.

It was proved in \cite{SV15} that balanced fairness is Pareto-efficient
when it is applied in polymatroid capacity sets,
which in practice yields explicit recursion formulas for the performance metrics.
However, if no further assumptions are made on the structure of the system,
the time complexity to compute these metrics is exponential with the number of queues.
It was proved in \cite{SV15} that it can be made linear
at the cost of strict assumptions
on the overall symmetry of the capacity set and the traffic intensity at each queue. 
Under symmetry in interaction across queues, it was shown in \cite{SV16} that the performance 
is robust to heterogeneity in loads and system configuration under an appropriate scaling regime. 
However, there is little understanding of performance for scenarios where queues themselves interact 
in heterogeneous fashion. 

In this paper, we consider a scenario
where the processor-sharing system is partitioned into a finite number of parts,
so that queues are exchangeable within each part and asymmetric across different parts.
For such systems, that we call \emph{poly-symmetric}, we obtain a performance expression with computational complexity which is polynomial in the number of queues.
We demonstrate the usefulness of these bounds by applying them to tree data networks, which are representative of backhaul networks,
and to randomly configured heterogeneous computer clusters.
In addition, we provide a monotonicity bound which allows us to bound performance of systems with capacity regions which are `nearly' poly-symmetric.

The paper is organized as follows.
Section \ref{sec:model} introduces the model
and shows that it applies to real systems as varied as
tree data networks and computer clusters.
We also recall known facts about balanced fairness.
In Section \ref{sec:polysymmetry}, we introduce the notion of poly-symmetry
and show that it yields explicit recursion formulas for the performance metrics
which have a complexity that is polynomial in the number of queues in the processor-sharing system.
Finally, Section \ref{sec:bounds} gives stochastic bounds
to compare the performance of different systems.
We conclude in Section \ref{sec:ccl}.

\section{System model}
\label{sec:model}

\subsection{Processor-sharing queueing system with a polymatroid capacity set}
\label{subsec:capacityset}

We consider a system of $n$ processor-sharing queues with coupled service rates
and we denote by $I = \{1,\ldots,n\}$ the set of queue indices.
For each $i \in I$, jobs enter the system at queue $i$
according to some Poisson process with intensity $\lambda_i$
and have i.i.d.\ exponential service requirements with mean $\sigma_i$,
resulting in a traffic intensity $\rho_i = \lambda_i \sigma_i$ at queue $i$.
Jobs leave the system immediately after service completion.
Such a queueing system will be called a \emph{processor-sharing system} throughout the paper.

The system state is described by the vector $x = (x_i : i \in I)$,
where $x_i$ is the number of jobs at queue $i$ for each $i \in I$.
For each $x \in \mathbb{N}^n$, $I(x) = \{i \in I: x_i > 0\}$
denotes the set of active queues in state $x$.
Queues have state-dependent service rates.
For each $x \in \mathbb{N}^n$, $\phi(x) = (\phi_i(x) : i \in I)$ denotes the vector of service rates per queue
when the system is in state $x$.

The system is characterized by a {\it capacity set},
which is defined as the set of all feasible resource allocations
$\phi = (\phi_i : i \in I) \in \mathbb{R}_+^n$.
This capacity set may be specified by practical constraints like
the capacities of the links in a data network
or the service rates of the servers in a computer cluster.
We are interested in queueing systems whose capacity set is
a particular type of polytope called a {\it polymatroid} \cite{F05}.

\begin{definition}
  \label{def:polymatroid}
  A polytope ${\cal C}$ in $\mathbb{R}_+^n$ is a {\it polymatroid} if there exists
  a non-negative function $\mu$ defined on the power set of $I$ such that
  $$
  {\cal C} =
  \left\{
    \phi \in \mathbb{R}_+^n:
    \sum_{i \in A} \phi_i \le \mu(A),
    \quad \forall A \subset I
  \right\}
  $$
  and $\mu$ satisfies the following properties:
  \begin{description}[noitemsep]
    \item Normalization: $\mu(\emptyset) = 0$,
    \item Monotonicity:
      for all $A, B \subset I$, if $A \subset B$, then $\mu(A) \le \mu(B)$,
    \item Submodularity:
      for all $A, B \subset I$,
      $\mu(A) + \mu(B) \ge \mu(A \cup B) + \mu(A \cap B)$.
  \end{description}
  $\mu$ is called the {\it rank function} of the polymatroid ${\cal C}$.
\end{definition}

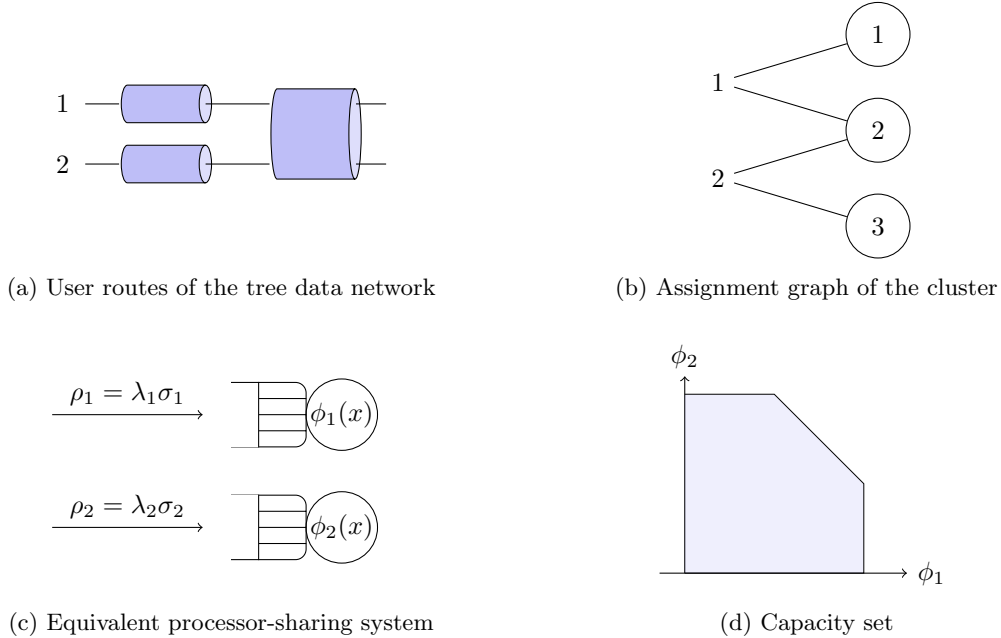
\begin{figure}
  \centering
  \begin{subfigure}[b]{.45\textwidth}
    \centering
    \begin{tikzpicture}
      \node (c1) {};
      \node at ($(c1)+(0,-.8cm)$) (c2) {};
      \node at ($(c1)!.5!(c2)+(2cm,0)$) (c) {};

      \node[
        cylinder,draw=black,aspect=0.7,
        minimum height=1.2cm,minimum width=.5cm,
        cylinder uses custom fill,
        cylinder body fill=myblue!40,
        cylinder end  fill=myblue!20
      ]
      at (c1) {};
      \node[
        cylinder,draw=black,aspect=0.7,
        minimum height=1.2cm,minimum width=.5cm,
        cylinder uses custom fill,
        cylinder body fill=myblue!40,
        cylinder end  fill=myblue!20
      ]
      at (c2) {};
      \node[
        cylinder,draw=black,aspect=0.7,
        minimum height=1.2cm,minimum width=1.2cm,
        cylinder uses custom fill,
        cylinder body fill=myblue!40,
        cylinder end  fill=myblue!20
      ]
      at (c) {};

      \node at ($(c1)+(-1.3cm,0)$) {$1$};
      \draw[-,rounded corners=.05cm]
      ($(c1)+(-1cm,0)$) -- ($(c1)+(-.55cm,0)$)
      ($(c1)+(.6cm,0)$) -- ($(c)+(-.55cm,.4cm)$)
      ($(c)+(.6cm,.4cm)$) -- ($(c)+(1cm,.4cm)$);

      \node at ($(c2)+(-1.3cm,0)$) {$2$};
      \draw[-,rounded corners=.05cm]
      ($(c2)+(-1cm,0)$) -- ($(c2)+(-.55cm,0)$)
      ($(c2)+(.6cm,0)$) -- ($(c)+(-.55cm,-.4cm)$)
      ($(c)+(.6cm,-.4cm)$) -- ($(c)+(1cm,-.4cm)$);
    \end{tikzpicture}
    \vspace{1cm}
    \caption{User routes of the tree data network}
    \label{fig:extree}
  \end{subfigure}
  ~
  \begin{subfigure}[b]{.45\textwidth}
    \centering
    \begin{tikzpicture}
      \def\scale{.85}

      \node[draw,circle,minimum size=\scale*1cm] (server1) {1};
      \node[draw,circle,minimum size=\scale*1cm]
      at ($(server1)+(0,-\scale*1.5cm)$) (server2) {2};
      \node[draw,circle,minimum size=\scale*1cm]
      at ($(server2)+(0,-\scale*1.5cm)$) (server3) {3};

      \node at ($(server1)!.5!(server2)+(-\scale*2.5cm,0)$) (class1) {1};
      \node at ($(server2)!.5!(server3)+(-\scale*2.5cm,0)$) (class2) {2};

      \draw[-] (class1) -- (server1);
      \draw[-] (class1) -- (server2);
      \draw[-] (class2) -- (server2);
      \draw[-] (class2) -- (server3);
    \end{tikzpicture}
    \caption{Assignment graph of the cluster}
    \label{fig:exmgraph}
  \end{subfigure}
  \\[.4cm]
  \centering
  \begin{subfigure}[b]{.45\textwidth}
    \centering
    \begin{tikzpicture}
      \node[draw,circle,minimum size=.95cm] (server1) {};
      \node at (server1) {$\phi_1(x)$};
      \node[draw,circle,minimum size=.95cm]
      at ($(server1)+(0,-1.5cm)$) (server2) {};
      \node at (server2) {$\phi_2(x)$};

      \node[queueps] (queue1) at ($(server1)+(-1.05cm,0)$) {};
      \fill[white] ([xshift=-\pgflinewidth-1pt,yshift=-\pgflinewidth+1pt] queue1.north west)
      rectangle ([xshift=\pgflinewidth+4pt,yshift=\pgflinewidth-1pt] queue1.south west);
      \fill[white] ([xshift=-\pgflinewidth-1pt,yshift=-\pgflinewidth] queue1.north west)
      rectangle ([xshift=15pt,yshift=\pgflinewidth] queue1.south west);
      \draw[-] ([xshift=15pt] queue1.north west)
      -- ([xshift=15pt] queue1.south west);

      \node[queueps] (queue2) at ($(server2)+(-1.05cm,0)$) {};
      \fill[white] ([xshift=-\pgflinewidth-1pt,yshift=-\pgflinewidth+1pt] queue2.north west)
      rectangle ([xshift=\pgflinewidth+4pt,yshift=\pgflinewidth-1pt] queue2.south west);
      \fill[white] ([xshift=-\pgflinewidth-1pt,yshift=-\pgflinewidth] queue2.north west)
      rectangle ([xshift=15pt,yshift=\pgflinewidth] queue2.south west);
      \draw[-] ([xshift=15pt] queue2.north west)
      -- ([xshift=15pt] queue2.south west);

      \draw[->] ($(queue1)+(-2.8cm,0)$)
      -- node[midway,above] {$\rho_1 = \lambda_1 \sigma_1$}
      ($(queue1)+(-.8cm,0)$);
      \draw[->] ($(queue2)+(-2.8cm,0)$)
      -- node[midway,above] {$\rho_2 = \lambda_2 \sigma_2$}
      ($(queue2)+(-.8cm,0)$);
    \end{tikzpicture}
    \vspace{.4cm}
    \caption{Equivalent processor-sharing system}
    \label{fig:expsnet}
  \end{subfigure}
  ~
  \begin{subfigure}[b]{.45\textwidth}
    \centering
    \begin{tikzpicture}
      \begin{axis}[
          axis equal,
          axis lines=middle,
          axis line style={->},
          xlabel={$\phi_1$}, xtick=\empty,
          every axis x label/.style={at={(ticklabel* cs:1.1)}},
          ylabel={$\phi_2$}, ytick=\empty,
          every axis y label/.style={at={(ticklabel* cs:1.1)}},
          xmin = 0, xmax = 2.2,
          ymin = 0, ymax = 2.2,
          height=4.2cm,
        ]
        \addplot[line width=.01cm,fill=myblue,fill opacity=0.1] coordinates{(0,0) (2,0) (2,1) (1,2) (0,2)};
      \end{axis}
    \end{tikzpicture}
    \caption{Capacity set}
    \label{fig:expolymatroid}
  \end{subfigure}
  \caption{A tree data network and a computer cluster with their representation as a processor-sharing system with $n = 2$ queues}
  \label{fig:exmodel}
\end{figure}

Before we specify the resource allocation,
we give two examples of real systems that fit into this model.

\subsection{Tree data networks}
\label{subsec:modeltree}

The first example is a data network with a tree topology \cite{BV04},
representative of backhaul networks of Internet service providers.
There are $n$ users that can generate flows in parallel
and we denote by $I = \{1,\ldots,n\}$ the set of user indices.
For any $i \in I$, user $i$ generates data flows according to some Poisson process with intensity $\lambda_i$
that is independent of the other users.
All flows generated by user $i$ follow the same route in the network and
have i.i.d.\ exponentially distributed sizes with mean $\sigma_i$ in bits,
resulting in a traffic intensity $\rho_i = \lambda_i \sigma_i$ in bit/s.
The state of the network is described by the vector $x = (x_i : i \in I)$,
where $x_i$ is the number of ongoing flows of user $i$, for each $i \in I$.

We make the following assumptions on the allocation of the resources.
The capacity of each link can be divided continuously among the flows that cross it.
Also, the resource allocation per flow only depends on the number of flows of each user in progress.
In particular, all flows of a user receive the same capacity,
so that the per-flow resource allocation is entirely defined in any state $x \in \mathbb{N}^n$
by the total capacity $\phi_i(x)$ allocated to flows of user $i$,
for any $i \in I$.

Under these assumptions, we can represent the data network
by a processor-sharing system with $n$ queues, one per user.
For each $i \in I$,
the jobs at queue $i$ in the equivalent processor-sharing system
are the ongoing flows of user $i$ in the data network,
and the service rate of this queue in state $x$ is the total capacity $\phi_i(x)$
allocated to the flows of user $i$.
We will now describe the corresponding capacity set.

Each link can be identified by the set of users that cross it.
Specifically, we can describe the network by a family ${\cal T}$ of subsets of $I$,
where a set $L \subset I$ is in ${\cal T}$
if and only if there is a link crossed by the flows of all users $i \in L$.
We assume that the network is a tree in the following way.

\begin{definition}
  The network is called a {\it tree} if for all $L, M \in {\cal T}$,
  $L \cap M \neq \emptyset$ implies that $L \subset M$ or $M \subset L$.
\end{definition}

\noindent There is no loss of generality in assuming that $I \in {\cal T}$,
for if not, the network is a forest where each subtree can be considered independently.
For each $L \in {\cal T}$, we denote by $C_L$ the capacity in bit/s of link $L$.
We assume that all links are constraining since otherwise we can simply ignore the non-constraining ones.
The resource allocation must then satisfy the capacity constraints
\begin{equation}
  \label{eq:treeconstraints}
  \sum_{i \in L} \phi_i(x) \le C_L,
  \quad \forall L \in {\cal T},
  \quad \forall x \in \mathbb{N}^n,
\end{equation}
so that the capacity set is given by
$$
{\cal C} = \left\{
  \phi \in \mathbb{R}_+^n:
  \sum_{i \in L} \phi_i \le C_L,
  \quad \forall L \in {\cal T}
\right\}.
$$

\begin{example}
  Figures \ref{fig:extree}, \ref{fig:expsnet} and \ref{fig:expolymatroid}
  give the example of a tree data network with $2$ users.
  The routes of the users are given in Figure \ref{fig:extree}.
  The flows of each user cross one link that is individual and another that is shared by both users.
  The representation of this data network as a processor-sharing system is given in Figure \ref{fig:expsnet}
  and the corresponding capacity set is given in Figure \ref{fig:expolymatroid}.
  It is easy to see that it is a polymatroid for any value of the link capacities.
\end{example}

The following theorem generalizes this last remark to any tree data network.

\begin{theorem}
  \label{theo:tree}
  The capacity set of a tree data network is a polymatroid
  with rank function $\mu$ defined by
  $$
  \mu(A)
  = \min
  \left\{
    \sum_{L \in \Sigma} C_L:
    \Sigma \subset {\cal T}
    \text{ is a family of disjoints sets s.t. }
    A \subset \bigcup_{L \in \Sigma} L
  \right\}
  $$
  for all non-empty set $A \subset I$.
  In addition, we have $\mu(L) = C_L$ for each $L \in {\cal T}$.
\end{theorem}

\begin{proof}
  We can certainly assume that ${\cal T}$ contains all the singletons
  since letting $C_{\{i\}} = \min_{L \subset {\cal T}, i \in L} C_L$ for each $i \in I$
  does not modify the capacity set ${\cal C}$.
  We can easily see that the result remains true if we do not make this assumption.

  We apply the following lemma which is a direct consequence of Theorems 2.5 and 2.6 of \cite{F05}
  about intersecting-submodular functions on intersecting families of subsets.

  \begin{lemma}
    Let ${\cal T}$ be a family of subsets of $I$
    and $g : {\cal T} \to \mathbb{R}$ such that,
    for all $L, M \in {\cal T}$ with $L \cap M \neq \emptyset$,
    we have $L \cap M \in {\cal T}$, $L \cup M \in {\cal T}$ and
    $g(L) + g(M) \ge g(L \cup M) + g(L \cap M)$.
    Further assume that $\emptyset, I \in {\cal T}$, $g(\emptyset) = 0$
    and ${\cal T}$ contains all the singletons of $I$.
    Then the set of solutions in $\mathbb{R}^n$ of the equations
    $$
    \sum_{i \in L} \phi_i \le g(L), \quad \forall L \in {\cal T},
    $$
    is given by
    $$
    \left\{
      \phi \in \mathbb{R}^n:
      \sum_{i \in A} \phi_i \le f(A), \quad \forall A \in I
    \right\},
    $$
    where $f$ is the real-valued, normalized, submodular function defined on the power set of $I$ by
    $$
    f(A) = \min\left\{
      \sum_{L \in \Sigma} g(L):
      \Sigma \subset {\cal T} \text{ is a partition of } A
    \right\},
    \quad \forall A \subset I.
    $$
  \end{lemma}

  \noindent
  The definition of a tree ensures that
  ${\cal T} \cup \{\emptyset\}$ satisfies the assumptions of the lemma,
  with the function $g$ defined on ${\cal T} \cup \{\emptyset\}$ by $g(L) = C_L$ for any $L \in {\cal T}$ and $g(\emptyset) = 0$.
  Hence, the set of solutions of the capacity constraints \eqref{eq:treeconstraints} in $\mathbb{R}^n$ is
  $$
  {\cal P} =
  \left\{
    \phi \in \mathbb{R}^n :
    \sum_{i \in A} \phi_i \le f(A),
    \quad \forall A \subset I
  \right\}
  $$
  where $f$ is the normalized, submodular function given by
  $$
  f(A)
  = \min \left\{
    \sum_{ L \in \Sigma} C_L:
    \Sigma \subset {\cal T} \text{ is a partition of } A
  \right\},
  \quad \forall A \subset I.
  $$
  Note that no claim about the monotonicity of $f$ can be made above
  because the points in ${\cal P}$ can have negative components.
  This is illustrated in Figure \ref{fig:treeset},
  where the intersection point of the sides of ${\cal P}$
  corresponding to the sets $\{1\}$ and $\{1,2\}$ has a negative ordinate because $f(\{1,2\}) < f(\{1\})$.

  \begin{figure}
    \centering
    \begin{tikzpicture}
      \begin{axis}[
          axis equal,
          axis lines=middle,
          axis line style={->},
          xmin = -1.1, xmax = 5.4,
          ymin = -.2, ymax = 1.2,
          xlabel={$\phi_1$}, xtick=\empty,
          every axis x label/.style={at={(ticklabel* cs:1.05)}},
          ylabel={$\phi_2$}, ytick=\empty,
          every axis y label/.style={at={(ticklabel* cs:1.05)}},
          height=6cm,
        ]
        \addplot[draw=none,fill=myred,fill opacity=0.08] coordinates{(-1,-1.8) (-1,2) (0,2) (0,0) (3,0) (4,-1) (4,-1.8)};
        \addplot[draw=none,fill=myblue,fill opacity=0.2] coordinates{(0,0) (0,2) (1,2) (3,0)};
        \addplot[-,myblue] coordinates{(0,0) (0,2) (1,2) (3,0) (0,0)};
        \addplot[densely dashdotted,myred] coordinates{(-1,2) (1,2) (4,-1) (4,-1.8)};

        \addplot[dashed] coordinates{(4.5,-1.5) (4,-1)};
        \addplot[dashed] coordinates{(1,2) (.5,2.5)};
        \node[align=center] at (axis cs:2.9,2.35) {$\mu(\{1,2\}) = f(\{1,2\})$};
        \addplot[dashed] coordinates{(4,.5) (4,-1)};
        \node[anchor=north] at (axis cs:4.7,-.13) {$f(\{1\})$};

        \addplot[dashed] coordinates{(3,.5) (3,-.4)};
        \node[align=center,anchor=north] at (axis cs:2.3,-.13) {$\mu(\{1\})$ \\ $= f(\{1,2\})$};

        \node[text=myred] at (axis cs:-.5,2.35) {${\cal P}$};
        \node[text=myblue] at (axis cs:1,1) {${\cal C}$};
      \end{axis}
    \end{tikzpicture}
    \caption{Construction of the capacity set of a tree data network in $\mathbb{R}_+^n$
    from the set of solutions of its capacity constraints in $\mathbb{R}^n$}
    \label{fig:treeset}
  \end{figure}
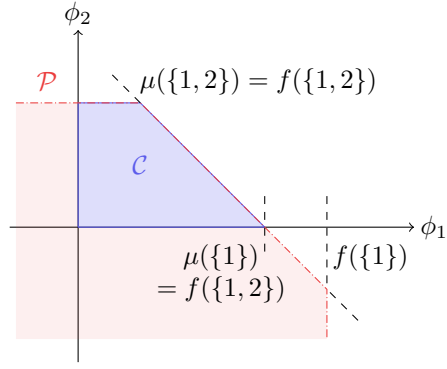

  Since the components of a vector of resource allocation are always positive,
  the capacity set ${\cal C}$ of the data network is given by ${\cal C} = {\cal P} \cap \mathbb{R}_+^n$.
  As we will see, since we restrict ourselves to points with positive components,
  the function $\mu$ which characterizes ${\cal C}$
  is not only normalized and submodular like $f$ but also non-decreasing.
  This is illustrated in Figure \ref{fig:treeset}, which shows that we can replace $f(\{1\})$ by $f(\{1,2\})$
  to describe the side corresponding to the set $\{1\}$ in ${\cal C}$.

  More formally, we prove that ${\cal C}$ is equal to the polymatroid ${\cal C}^\prime$ with rank function $\mu$ given by
  $$
  \mu(A)
  = \min \left\{
    f(B):
    A \subset B \subset I
  \right\},
  \quad \forall A \subset I.
  $$
  One can check that this function $\mu$ coincides with the one given in the theorem statement.
  We first show that $\mu$ is indeed a rank function
  and then we prove that ${\cal C}$ is equal to ${\cal C}^\prime$.

  The normalization of $\mu$ follows from that of $f$.
  Also $\mu$ is non-decreasing by construction.
  Finally, for each $A, B \subset I$,
  we have
  $\mu(A) + \mu(B) = f(A^\prime) + f(B^\prime)$
  for some $A^\prime, B^\prime \subset I$
  such that $A \subset A^\prime$ and $B \subset B^\prime$,
  and also
  $$
  f(A^\prime) + f(B^\prime)
  \ge f(A^\prime \cup B^\prime)
  + f(A^\prime \cap B^\prime)
  \ge \mu(A \cup B) + \mu(A \cap B),
  $$
  where the first inequality holds by submodularity of $f$
  and the second by definition of $\mu$,
  since $A \cup B \subset A^\prime \cup B^\prime$
  and $A \cap B \subset A^\prime \cap B^\prime$.
  Hence $\mu$ is submodular.

  We finally prove that ${\cal C} = {\cal C}^\prime$.
  It is clear that any vector in ${\cal C}^\prime$ is also in ${\cal C} = {\cal P} \cap \mathbb{R}_+^n$
  since $\mu(A) \le f(A)$ for all $A \subset I$.
  Conversely, consider $\phi \in {\cal C}$.
  If $\phi$ is not in ${\cal C}^\prime$, then there is $A \subset I$ so that
  $\sum_{i \in A} \phi_i > \mu(A)$,
  which implies that $\mu(A) < f(A)$.
  By definition of $\mu$, it follows that
  there is $B \subset I$ so that $A$ is a strict subset of $B$ and $f(B) = \mu(A)$.
  But then
  $$
  \sum_{i \in B \setminus A} \phi_i
  = \sum_{i \in B} \phi_i - \sum_{i \in A} \phi_i
  < f(B) - \mu(A) = 0,
  $$
  so that at least one component of $\phi$ is negative. This is a contradiction.
\end{proof}

\begin{example}
  \begin{figure}
    \centering
    \begin{subfigure}[b]{.48\textwidth}
      \centering
      \begin{tikzpicture}
        \node (c1) {};
        \node at ($(c1)+(0,-.8cm)$) (c2) {};
        \node at ($(c1)!.5!(c2)+(1.8cm,0)$) (c12) {};
        \node at ($(c2)+(1.8cm,-.8cm)$) (c3) {};
        \node at ($(c12)!.5!(c3)+(1.8cm,0)$) (c) {};

        \node[
          cylinder,draw=black,aspect=0.7,
          minimum height=1.2cm,minimum width=.5cm,
          cylinder uses custom fill,
          cylinder body fill=myblue!40,
          cylinder end  fill=myblue!20
        ]
        at (c1) {};
        \node[
          cylinder,draw=black,aspect=0.7,
          minimum height=1.2cm,minimum width=.5cm,
          cylinder uses custom fill,
          cylinder body fill=myblue!40,
          cylinder end  fill=myblue!20
        ]
        at (c2) {};
        \node[
          cylinder,draw=black,aspect=0.7,
          minimum height=1.2cm,minimum width=.5cm,
          cylinder uses custom fill,
          cylinder body fill=myblue!40,
          cylinder end  fill=myblue!20
        ]
        at (c3) {};
        \node[
          cylinder,draw=black,aspect=0.7,
          minimum height=1.2cm,minimum width=.7cm,
          cylinder uses custom fill,
          cylinder body fill=myblue!40,
          cylinder end  fill=myblue!20
        ]
        at (c12) {};
        \node[
          cylinder,draw=black,aspect=0.7,
          minimum height=1.2cm,minimum width=.8cm,
          cylinder uses custom fill,
          cylinder body fill=myblue!40,
          cylinder end  fill=myblue!20
        ]
        at (c) {};

        \node at ($(c1)+(-1.1cm,0)$) {$1$};
        \draw[-,rounded corners=.05cm]
        ($(c1)+(-.8cm,0)$) -- ($(c1)+(-.55cm,0)$)
        ($(c1)+(.6cm,0)$) -- ($(c12)+(-.55cm,.05cm)$)
        ($(c12)+(.6cm,.05cm)$) -- ($(c)+(-.55cm,.1cm)$)
        ($(c)+(.6cm,.1cm)$) -- ($(c)+(.8cm,.1cm)$);

        \node at ($(c2)+(-1.1cm,0)$) {$2$};
        \draw[-,rounded corners=.05cm]
        ($(c2)+(-.8cm,0)$) -- ($(c2)+(-.55cm,0)$)
        ($(c2)+(.6cm,0)$) -- ($(c12)+(-.55cm,-.05cm)$)
        ($(c12)+(.6cm,-.05cm)$) -- ($(c)+(-.55cm,0)$)
        ($(c)+(.6cm,0)$) -- ($(c)+(.8cm,0)$);

        \node at ($(c3)+(-1.1cm,0)$) {$3$};
        \draw[-,rounded corners=.05cm]
        ($(c3)+(-.8cm,0)$) -- ($(c3)+(-.55cm,0)$)
        ($(c3)+(.6cm,0)$) -- ($(c)+(-.55cm,-.1cm)$)
        ($(c)+(.6cm,-.1cm)$) -- ($(c)+(.8cm,-.1cm)$);
      \end{tikzpicture}
      \vspace{.5cm}
      \caption{User routes}
      \label{fig:extreegeneralgraph}
    \end{subfigure}
    \qquad
    \begin{subfigure}[b]{.42\textwidth}
      \centering
      \begin{tikzpicture}
        \begin{axis}[
            view={125}{13},
            axis equal,
            axis lines=middle,
            axis line style={->},
            xlabel={$\phi_1$}, xtick=\empty,
            every axis x label/.style={at={(ticklabel* cs:1.1)}},
            ylabel={$\phi_2$}, ytick=\empty,
            every axis y label/.style={at={(ticklabel* cs:1.1)}},
            zlabel={$\phi_3$}, ztick=\empty,
            every axis z label/.style={at={(ticklabel* cs:1.1)}},
            xmin = 0, xmax = 2.2,
            ymin = 0, ymax = 2.2,
            zmin = 0, zmax = 2.2,
            height=5.5cm,
          ]

          \addplot3[line width=.01cm,fill=myblue,fill opacity=0.1]
          coordinates{(2,0,0) (2,1,0) (2,1,1) (2,0,2) (2,0,0)};
          \addplot3[line width=.01cm,fill=myblue,fill opacity=0.1]
          coordinates{(0,2,0) (1,2,0) (1,2,1) (0,2,2) (0,2,0)};
          \addplot3[line width=.01cm,fill=myblue,fill opacity=0.1]
          coordinates{(0,0,2) (2,0,2) (0,2,2) (0,0,2)};

          \addplot3[draw=none,fill=myblue,fill opacity=0.1]
          coordinates{(2,1,1) (1,2,1) (0,2,2) (2,0,2)};
          \addplot3[draw=none,fill=myblue,fill opacity=0.1]
          coordinates{(2,1,0) (1,2,0) (1,2,1) (2,1,1)};
          \addplot3[line width=.01cm] coordinates{(2,1,0) (1,2,0)};
          \addplot3[line width=.01cm] coordinates{(1,2,1) (2,1,1)};
        \end{axis}
      \end{tikzpicture}
      \caption{Capacity set}
      \label{fig:extreegeneralregion}
    \end{subfigure}
    \caption{Representation of a tree data network}
    \label{fig:extreegeneral}
  \end{figure}
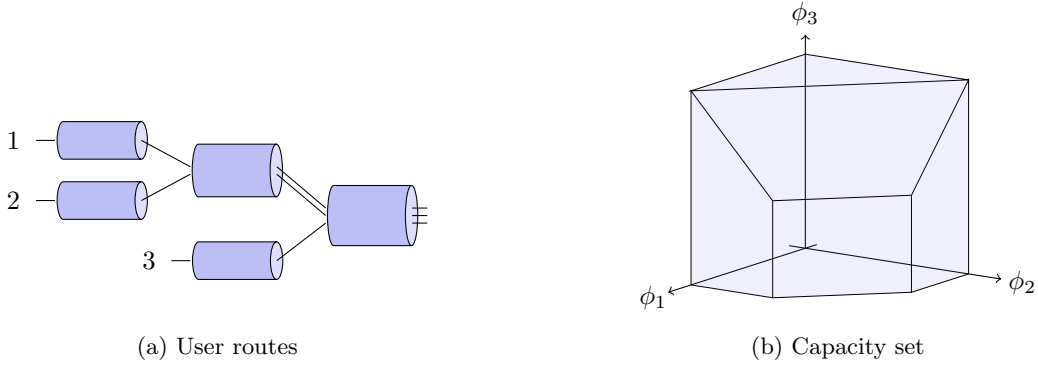
  \label{ex:tree}
  Figure \ref{fig:extreegeneral} gives the example of a tree data network with its capacity set.
  The routes of the users are given in Figure \ref{fig:extreegeneralgraph}.
  Each link is labeled with the set of user indices whose flows cross this link.
  The capacity constraints are
  $$
  \phi_1 \le C_{\{1\}}, \quad
  \phi_2 \le C_{\{2\}}, \quad
  \phi_3 \le C_{\{3\}}, \quad
  \phi_1 + \phi_2 \le C_{\{1,2\}}, \quad
  \phi_1 + \phi_2 + \phi_3 \le C_{\{1,2,3\}}.
  $$
  The rank function $\mu$ of the capacity set is given by
  \begin{align*}
    \mu(\{1\}) &= C_{\{1\}}, ~
    &\mu(\{1,2\}) &= C_{\{1,2\}}, ~
    &\mu(\{1,2,3\}) &= C_{\{1,2,3\}}. \\
    \mu(\{2\}) &= C_{\{2\}}, ~
    &\mu(\{1,3\}) &= \min\left( C_{\{1\}} + C_{\{3\}}, C_{\{1,2,3\}} \right), \\
    \mu(\{3\}) &= C_{\{3\}}, ~
    &\mu(\{2,3\}) &= \min\left( C_{\{2\}} + C_{\{3\}}, C_{\{1,2,3\}} \right),
  \end{align*}
\end{example}

\subsection{Computer clusters}
\label{subsec:modelcluster}

We consider a cluster of $m$ servers
which can be pooled to process jobs in parallel.
The set of servers is denoted by $S = \{1,\ldots,m\}$.
There are $n$ classes of jobs and we denote by $I = \{1,\ldots,n\}$ the set of class indices.
For any $i \in I$, class-$i$ jobs enter the cluster as a Poisson process with intensity $\lambda_i$
and have i.i.d.\ exponential service requirements with mean $\sigma_i$,
resulting in a traffic intensity $\rho_i = \lambda_i \sigma_i$ for class $i$.
Jobs leave the cluster immediately after service completion.
The state of the cluster is described by the vector $x = (x_i : i \in I)$,
where $x_i$ is the number of jobs of class $i$, for each $i \in I$.

The class of a job defines the set of servers that can process it.
The server assignment is given by a family
$(S_i : i \in I)$ of subsets of $S$,
where $S_i$ denotes the set of servers that can serve class-$i$ jobs,
for each $i \in I$.
Equivalently, the server assignment can be described by a bipartite graph
$$
G = \left( I, S, \bigcup_{i \in S} ( \{i\} \times S_i ) \right)
$$
called the {\it assignment graph} of the computer cluster.
The service capacity of server $s$ is $\mu_s$ for each $s = 1,\ldots,m$.
For any set $A \subset I$ of job classes, we let
\begin{equation}
  \label{eq:mucluster}
  \mu(A) = \sum_{s \in \bigcup_{i \in A} S_i} \mu_s
\end{equation}
denote the aggregate capacity available for the classes in $A$.

We make the following assumptions on the allocation of the server capacities.
Servers can be pooled to process jobs in parallel.
When a job is in service on several servers, its service rate is the sum of the rates allocated by each server to this job.
We also assume that the capacity of each server can be divided continuously among the jobs it can serve.
Finally, the allocation of the service rates per job only depends on the number of jobs of each class in the cluster.
In particular, all jobs of a class receive service at the same rate,
so that the per-job resource allocation is entirely defined in any state $x \in \mathbb{N}^n$
by the total capacity $\phi_i(x)$ allocated to class-$i$ jobs, for each $i \in I$.

Under these assumptions, we can describe the evolution of the cluster with a
processor-sharing system with $n$ queues, one per class.
For each $i \in I$, queue $i$ contains class-$i$ jobs
and its service rate in state $x$ is the total capacity $\phi_i(x)$ allocated to class-$i$ jobs collectively.
It was proved in \cite{SV15} that the capacity set of such a cluster is a polymatroid
and that the function $\mu$ defined by \eqref{eq:mucluster} is its rank function.

\begin{example}
  Figure \ref{fig:exmgraph} gives the assignment graph for an example of a computer cluster,
  where job classes are on the left and servers are on the right.
  Server $2$ can serve both classes whereas servers $1$ and $3$ are specialized.
  The corresponding processor-sharing system with $2$ queues is shown in Figure \ref{fig:expsnet}
  and its capacity set, which is a polymatroid in $\mathbb{R}_+^2$, is depicted Figure \ref{fig:expolymatroid}.
  The vertical and horizontal sides correspond to the individual constraints
  of classes $1$ and $2$,
  with $\mu(\{1\}) = \mu_1 + \mu_2$
  and $\mu(\{2\}) = \mu_2 + \mu_3$.
  The diagonal side corresponds to the joint constraint on classes $1$ and $2$,
  with $\mu(\{1,2\}) = \mu_1 + \mu_2 + \mu_3$.
\end{example}

\subsection{Balanced fairness}
\label{subsec:bf}

The service rates are allocated by applying balanced fairness \cite{BP03-1}
in the polymatroid capacity set ${\cal C}$ introduced in Section \ref{subsec:capacityset}.

For each $i \in I$, let $e_i$ denote the $n$-dimensional vector with $1$ in position $i$ and $0$ elsewhere.
Balanced fairness is defined as the only resource allocation that both satisfies the {\it balance property}
$$
\phi_i(x) \phi_j(x-e_i)
= \phi_i(x-e_j) \phi_j(x),
\quad \forall x \in \mathbb{N}^n,
\quad \forall i,j \in I(x),
$$
and maximizes the resource utilization in the following sense:
in any state $x \in \mathbb{N}^n$, $\phi(x) \in {\cal C}$ and there exists $A \subset I(x)$ such that
$$
\sum_{i \in A} \phi_i(x) = \mu(A).
$$
The balance property ensures that there exists a balance function $\Phi$ on $\mathbb{N}^n$ such that
$\Phi(0) = 1$ and
$$
\phi_i(x) = \frac{\Phi(x-e_i)}{\Phi(x)},
\quad \forall x \in \mathbb{N}^n \setminus \{0\},
\quad \forall i \in I(x).
$$
The second condition implies that $\Phi$ satisfies the recursion
$$
\Phi(x) = \max_{A \subset I(x)}
\left\{
  \frac{\sum_{i \in A} \Phi(x-e_i)}{\mu(A)}
\right\},
\quad \forall x \in \mathbb{N}^n \setminus \{0\}.
$$
In \cite{SV15} it is proved that balanced fairness is Pareto-efficient in polymatroid capacity sets,
which means that this maximum is always achieved by the set $I(x) = \{i \in I : x_i > 0\}$ of active queues:
\begin{equation}
  \label{eq:recPhi}
  \Phi(x)
  = \frac{\sum_{i \in I(x)} \Phi(x-e_i)}{\mu(I(x))},
  \quad \forall x \in \mathbb{N}^n \setminus \{0\}.
\end{equation}

Since the balance property is satisfied, the processor-sharing system defined in Section \ref{subsec:capacityset}
is a Whittle network \cite{S99}.
A stationary measure of the system state ${\bf X} = ({\bf X}_i : i \in I)$ is
$$
\pi(x) = \pi(0) \Phi(x) \rho^x,
\quad \forall x \in \mathbb{N}^n,
$$
where we use the notation $\rho^x = \prod_{i \in I} {\rho_i}^{x_i}$ for any $x \in \mathbb{N}^n$.
Substituting (\ref{eq:recPhi}) into this expression yields
$$
\pi(x) = \frac{\sum_{i \in I(x)} \rho_i \pi(x-e_i)}{\mu(I(x))},
\quad \forall x \in \mathbb{N}^n \setminus \{0\}.
$$
It is proved in \cite{BP03-1} that the system is stable,
in the sense that the underlying Markov process is ergodic,
if and only if
$$
\sum_{i \in A} \rho_i < \mu(A),
\quad \forall A \subset I,
$$
which means that the vector of traffic intensities belongs to the interior of the capacity set.
In the rest of the paper, we assume that this condition is satisfied
and we denote by $\pi$ the stationary distribution of the system state.

\subsection{Performance metrics}
\label{subsec:perf}

By abuse of notation, for each $A \subset I$,
we denote by $\pi(A)$ the stationary probability that the set of active queues is $A$:
$$
\pi(A)
= \mathbb{P}\left\{ I({\bf X}) = A \right\}
= \sum_{\substack{x \in \mathbb{N}^n, \\ I(x) = A}} \pi(x),
\quad \forall A \subset I.
$$
For each $i \in I$,
let $L_i = \mathbb{E}[{\bf X}_i]$ denote the mean number of jobs at queue $i$ and,
for each $A \subset I$,
let $L_i(A) = \mathbb{E}\left[ {\bf X}_i | I({\bf X}) = A \right]$
denote the mean number of jobs at queue $i$ given that the set of active queues is $A$.
By the law of total expectation, we have
$$
L_i
= \sum_{A \subset I} L_i(A) \pi(A),
\quad \forall i \in I.
$$
The following theorem gives a recursive formula for $\pi(A)$ and $L_i(A)$
for any $A \subset I$ and $i \in I$.
It is a restatement of Theorem $4$ in \cite{SV15}
using the same idea as Proposition $4$ and Theorem $1$ in \cite{SV16}.
\begin{theorem}
  \label{theo:recsets}
  For each non-empty set $A \subset I$, we have
  \begin{equation}
    \label{eq:recpisets}
    \pi(A)
    = \frac{\sum_{i \in A} \rho_i \pi(A \setminus \{i\})}
    {\mu(A) - \sum_{i \in A} \rho_i}.
  \end{equation}
  Let $i \in I$.
  For each set $A \subset I$,
  we have $L_i(A) = 0$ if $i \notin A$,
  and otherwise
  \begin{equation}
    \label{eq:recLsets}
    \pi(A) L_i(A)
    = \frac
    {
      \rho_i \pi(A \setminus \{i\}) + \rho_i \pi(A)
      + \sum_{j \in A \setminus \{i\}} \rho_j \pi(A \setminus \{j\}) L_i(A \setminus \{j\})
    }
    {\mu(A) - \sum_{j \in A} \rho_j}.
  \end{equation}
\end{theorem}

Observe that \eqref{eq:recpisets}
allows one to evaluate recursively
$\pi(A) / \pi(\emptyset)$ for each $A \subset I$,
from which $\pi(\emptyset)$ can be computed.
Similarly, for each $i \in I$, \eqref{eq:recLsets} allows one to evaluate recursively
$\pi(A) L_i(A) / \pi(\emptyset)$ for each $A \subset I$ and each $i \in I$,
from which the value of $L_i$ can be deduced.
One could then compute performance metrics
like the mean delay or the mean service rate per queue from $L_i$
by applying Little's law.
Note that the complexity is exponential in the number of queues.

\section{Poly-symmetry}
\label{sec:polysymmetry}

\subsection{Definition}
\label{subsec:defpoly}

The exponential complexity of the formulas of Theorem \ref{theo:recsets} makes it impractical
when we want to predict the performance of large-scale systems.
To cope with this, we introduce the notion of poly-symmetry,
which allows us to obtain formulas with a complexity that is polynomial in the number of queues
at the cost of some regularity assumptions on the capacity set and the traffic intensity at each queue.
Poly-symmetry is a generalization of the notion of symmetry which was considered in \cite{SV15,SV16}.

The following definition will be used subsequently to introduce poly-symmetry.
It is easy to check that it defines an equivalence relation on the set $I$ of indices.

\begin{definition}
  \label{def:exch}
  Let ${\cal C}$ be a polymatroid on $\mathbb{R}_+^n$ and denote its rank function by $\mu$.
  Let $i, j \in I$ with $i \neq j$.
  We say that indices $i$ and $j$ are {\it exchangeable} in ${\cal C}$ if
  $$
  \mu(A \cup \{i\})
  = \mu(A \cup \{j\}),
  \quad \forall A \subset I \setminus \{i,j\}.
  $$
\end{definition}

\noindent As the name suggests, two indices are exchangeable if and only if
exchanging these indices does not modify the capacity set.
Note that the exchangeability of two indices $i$ and $j$
implies that they have the same individual constraints $\mu(\{i\}) = \mu(\{j\})$.
The reverse implication is not true when $n > 2$, as we will see in the following example.

\begin{example}
  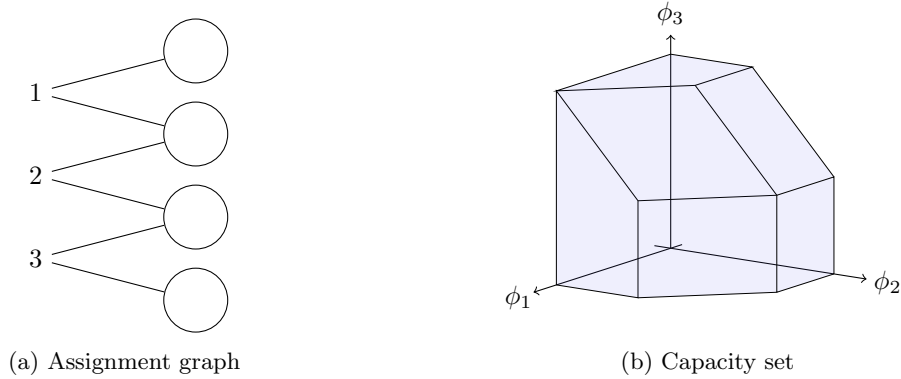
\begin{figure}
    \centering
    \begin{subfigure}[b]{.42\textwidth}
      \centering
      \begin{tikzpicture}
        \def\scale{.85}

        \node[draw,circle,minimum size=\scale*1cm] (server1) {};
        \node[draw,circle,minimum size=\scale*1cm]
        at ($(server1)+(0,-\scale*1.3cm)$) (server2) {};
        \node[draw,circle,minimum size=\scale*1cm]
        at ($(server2)+(0,-\scale*1.3cm)$) (server3) {};
        \node[draw,circle,minimum size=\scale*1cm]
        at ($(server3)+(0,-\scale*1.3cm)$) (server4) {};

        \node at ($(server1)!.5!(server2)+(-\scale*2.5cm,0)$) (class1) {1};
        \node at ($(server2)!.5!(server3)+(-\scale*2.5cm,0)$) (class2) {2};
        \node at ($(server3)!.5!(server4)+(-\scale*2.5cm,0)$) (class3) {3};

        \draw[-] (class1) -- (server1);
        \draw[-] (class1) -- (server2);
        \draw[-] (class2) -- (server2);
        \draw[-] (class2) -- (server3);
        \draw[-] (class3) -- (server3);
        \draw[-] (class3) -- (server4);
      \end{tikzpicture}
      \caption{Assignment graph}
      \label{fig:exexchgraph}
    \end{subfigure}
    \qquad
    \begin{subfigure}[b]{.42\textwidth}
      \centering
      \tdplotsetmaincoords{75}{127}
      \begin{tikzpicture}[tdplot_main_coords]
        \begin{axis}[
            view={125}{13},
            axis equal,
            axis lines=middle,
            axis line style={->},
            xlabel={$\phi_1$}, xtick=\empty,
            every axis x label/.style={at={(ticklabel* cs:1.1)}},
            ylabel={$\phi_2$}, ytick=\empty,
            every axis y label/.style={at={(ticklabel* cs:1.1)}},
            zlabel={$\phi_3$}, ztick=\empty,
            every axis z label/.style={at={(ticklabel* cs:1.1)}},
            xmin = 0, xmax = 2.2,
            ymin = 0, ymax = 2.2,
            zmin = 0, zmax = 2.2,
            height=5.5cm,
          ]
          \addplot3[line width=.01cm,fill=myblue,fill opacity=0.1] coordinates{(2,0,0) (2,0,2) (2,1,1) (2,1,0) (2,0,0)};
          \addplot3[line width=.01cm,fill=myblue,fill opacity=0.1] coordinates{(0,0,2) (2,0,2) (1,1,2) (0,1,2) (0,0,2)};

          \addplot3[draw=none,fill=myblue,fill opacity=0.1] coordinates{(2,1,0) (2,1,1) (1,2,1) (1,2,0) (2,1,0)};
          \addplot3[line width=.01cm] coordinates{(2,1,0) (1,2,0)};
          \addplot3[line width=.01cm] coordinates{(2,1,1) (1,2,1)};

          \addplot3[draw=none,fill=myblue,fill opacity=0.1] coordinates{(0,1,2) (1,1,2) (1,2,1) (0,2,1) (0,1,2)};
          \addplot3[line width=.01cm] coordinates{(0,1,2) (0,2,1)};
          \addplot3[line width=.01cm] coordinates{(1,1,2) (1,2,1)};

          \addplot3[draw=none,fill=myblue,fill opacity=0.1] coordinates{(2,0,2) (2,1,1) (1,2,1) (1,1,2) (2,0,2)};

          \addplot3[draw=none,fill=myblue,fill opacity=0.1] coordinates{(0,2,0) (1,2,0) (1,2,1) (0,2,1) (0,2,0)};
          \addplot3[line width=.01cm] coordinates{(0,2,0) (1,2,0) (1,2,1) (0,2,1) (0,2,0)};
        \end{axis}
      \end{tikzpicture}
      \vspace{.2cm}
      \caption{Capacity set}
      \label{fig:exexchset}
    \end{subfigure}
    \caption{Computer cluster with two exchangeable indices and a third index}
    \label{fig:exexch}
  \end{figure}
  \label{ex:exchindices}
  Consider the computer cluster with the assignment graph depicted in Figure \ref{fig:exexchgraph},
  where all servers have the same unit capacity.
  The corresponding polymatroid capacity set is illustrated in Figure \ref{fig:exexchset}.
  We have $\mu(\{1\}) = \mu(\{3\}) = 2$
  and $\mu(\{1,2\}) = \mu(\{2,3\}) = 3$,
  so that indices $1$ and $3$ are exchangeable.
  Index $2$ is not exchangeable with any of the two other indices because
  $\mu(\{1,2\}) = \mu(\{2,3\}) = 3$
  while $\mu(\{1,3\}) = 4$.
\end{example}

Let us now define poly-symmetry.
Suppose $K \ge 1$
and consider a partition $\Sigma = (I_k: k = 1,\ldots,K)$ of $I$ in $K$ parts.

\begin{definition}
  \label{def:poly}
  Let ${\cal C}$ be a polymatroid in $\mathbb{R}_+^n$.
  ${\cal C}$ is called {\it poly-symmetric} with respect to partition $\Sigma$ if
  for any $k = 1,\ldots,K$, all indices in $I_k$ are pairwise exchangeable in ${\cal C}$.
\end{definition}

\noindent
Since the exchangeability of indices defines an equivalence relation on $I$,
we can consider the quotient set of $I$ by this relation,
which is the partition of $I$ into the maximal sets of pairwise exchangeable indices.
Definition \ref{def:poly} can then be rephrased as follows:
a polymatroid ${\cal C}$ is poly-symmetric with respect to a partition $\Sigma$
if and only if $\Sigma$ is a refinement of the quotient set of $I$ by the exchangeability relation in ${\cal C}$.
It follows directly from the definition that the polymatroid of Example \ref{ex:tree}
is poly-symmetric with respect to partition $( \{1,2\}, \{3\})$
when $C_{\{1\}} = C_{\{2\}}$, as we can see in Figure \ref{fig:extreegeneralregion}.
Also in Example \ref{ex:exchindices}, the polymatroid
is poly-symmetric with respect to partition $( \{1,3\},\{2\} )$.

For each $k = 1,\ldots,K$,
let $n_k = |I_k|$ denote the size of part $k$,
where by part we mean a subset of the partition.
For any $A \subset I$, let
$|A|_\Sigma = (|A \cap I_k|: k = 1,\ldots,K)$
denote the vector of sizes of each part of $A$ in the partition.
The set of these vectors is denoted by
$$
{\cal N} = \prod_{k=1}^K \{0,1,\ldots,n_k\}.
$$
We now give an alternative definition of poly-symmetry
which is equivalent to Definition \ref{def:poly}.
It is a generalization of the definition of symmetry given in \cite{SV15,SV16}.
We will use it to express and prove Theorem \ref{theo:recpoly}.

\begin{definition}
  \label{def:poly2}
  Let ${\cal C}$ be a polymatroid in $\mathbb{R}_+^n$ and denote its rank function by $\mu$.
  ${\cal C}$ is called {\it poly-symmetric} with respect to partition $\Sigma$
  if for any $A \subset I$,
  $\mu(A)$ depends on $A$ only through the size of $A \cap I_k$ for each $k = 1,\ldots,K$.
  Equivalently, there exists a componentwise non-decreasing function $h: {\cal N} \to \mathbb{R}_+$
  such that $\mu(A) = h(|A|_\Sigma)$ for all $A \subset I$.
  We call $h$ the {\it cardinality rank function} of ${\cal C}$ with respect to partition $\Sigma$.
\end{definition}

\noindent {\it Proof of the equivalence.}\
  We only prove that Definition \ref{def:poly} implies Definition \ref{def:poly2};
  the reverse implication is clear.
  For any $A, B \subset I$ with $|A|_\Sigma = |B|_\Sigma$,
  we can write
  $A = (A \setminus B) \sqcup (A \cap B)$ and $B = (B \setminus A) \sqcup (A \cap B)$,
  where $\sqcup$ denotes the union of two disjoint sets.
  Since we have $|A \setminus B|_\Sigma = |B \setminus A|_\Sigma$,
  we are thus reduced to proving that $\mu(A \sqcup C) = \mu(B \sqcup C)$
  for all disjoint sets $A, B, C \subset I$ such that $|A|_\Sigma = |B|_\Sigma$.
  This can be done by ascending induction on the cardinality of $A$ and $B$.
\qed

\begin{example}
  \begin{figure}
    \centering
    \begin{subfigure}[b]{.42\textwidth}
      \centering
      \begin{tikzpicture}
        \def\scale{.85}

        \node[draw,circle,minimum size=\scale*1cm] (server1) {};
        \node[draw,circle,minimum size=\scale*1cm]
        at ($(server1)+(0,-\scale*1.5cm)$) (server2) {};
        \node[draw,circle,minimum size=\scale*1cm]
        at ($(server2)+(0,-\scale*1.5cm)$) (server3) {};

        \node at ($(server1)+(-\scale*2.5cm,0)$) (class1) {1};
        \node at ($(server2)+(-\scale*2.5cm,0)$) (class2) {2};
        \node at ($(server3)+(-\scale*2.5cm,0)$) (class3) {3};

        \draw[-] (class1) -- (server1);
        \draw[-] (class1) -- (server2);
        \draw[-] (class2) -- (server1);
        \draw[-] (class2) -- (server2);
        \draw[-] (class2) -- (server3);
        \draw[-] (class3) -- (server2);
        \draw[-] (class3) -- (server3);
      \end{tikzpicture}
      \caption{Assignment graph}
      \label{fig:expolysymgraph}
    \end{subfigure}
    \qquad
    \begin{subfigure}[b]{.42\textwidth}
      \centering
      \tdplotsetmaincoords{75}{127}
      \begin{tikzpicture}[tdplot_main_coords]
        \begin{axis}[
            view={125}{13},
            axis equal image,
            axis lines=middle,
            axis line style={->},
            xlabel={$\phi_1$}, xtick=\empty,
            every axis x label/.style={at={(ticklabel* cs:1.1)}},
            ylabel={$\phi_2$}, ytick=\empty,
            every axis y label/.style={at={(ticklabel* cs:1.1)}},
            zlabel={$\phi_3$}, ztick=\empty,
            every axis z label/.style={at={(ticklabel* cs:1.1)}},
            xmin = 0, xmax = 2.2,
            ymin = 0, ymax = 3.2,
            zmin = 0, zmax = 2.2,
            height=6cm,
          ]
          \addplot3[draw=none,fill=myblue,fill opacity=0.1] coordinates{(2,0,0) (2,0,1) (2,1,0) (2,0,0)};
          \addplot3[line width=.01cm] coordinates{(2,0,1) (2,0,0) (2,1,0)};

          \addplot3[draw=none,fill=myblue,fill opacity=0.1] coordinates{(0,0,2) (1,0,2) (0,1,2) (0,0,2)};
          \addplot3[line width=.01cm] coordinates{(1,0,2) (0,0,2) (0,1,2)};

          \addplot3[line width=.01cm,fill=myblue,fill opacity=0.1] coordinates{(2,0,1) (1,0,2) (0,1,2) (0,3,0) (2,1,0) (2,0,1)};

          \addplot3[myblue,fill=myblue,fill opacity=.35] coordinates{(0,0,0) (1.5,0,1.5) (0,3,0)};
          \addplot3[->,myblue] coordinates{(0,0,0) (.7,1.4,.7)}
          node[text=myblue,pos=.7,opacity=1,xshift=-.3cm,yshift=.03cm] {\bf $\rho$};
        \end{axis}
      \end{tikzpicture}
      \caption{Capacity set}
      \label{fig:expolysymregion}
    \end{subfigure}
    \caption{Computer cluster with a polymatroid capacity set
    which is poly-symmetric with respect to partition $\Sigma = ( \{1,3\}, \{2\} )$}
    \label{fig:expolysym}
  \end{figure}
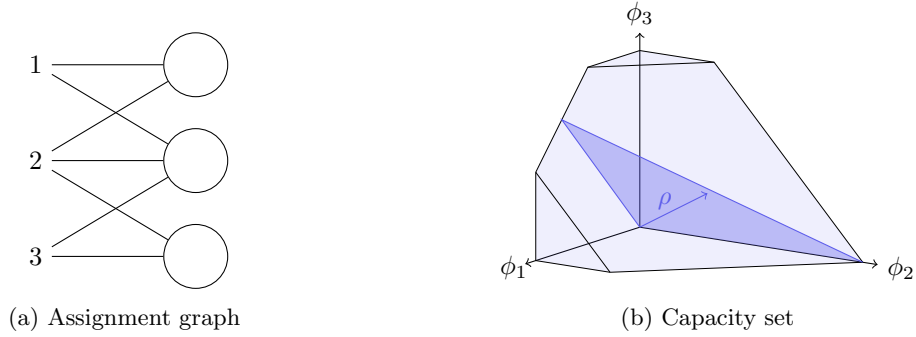
  \label{ex:poly}
  Consider the computer cluster with the assignment graph depicted in Figure \ref{fig:expolysym},
  where all servers have the same unit capacity.
  The corresponding capacity set is illustrated in Figure \ref{fig:expolysymregion}.
  It is poly-symmetric with respect to partition $\Sigma = ( \{1,3\}, \{2\} )$
  and the corresponding cardinality rank function $h$
  is given by $h(0,0) = 0$, $h(1,0) = 2$ and $h(0,1) = h(1,1) = 3$.
\end{example}

\subsection{Performance metrics}
\label{subsec:recpoly}

Let $\Sigma = (I_k: k = 1,\ldots,K)$ be a partition of $I$.
We consider a processor-sharing system with a polymatroid capacity set
which is poly-symmetric with respect to $\Sigma$.
For each $A \subset I$, the vector
$|A|_\Sigma = (|A \cap I_k|: k = 1,\ldots,K)$
gives the number of active queues in each part of the partition
when the set of active queues is $A$.
By abuse of notation, for each $k = 1,\ldots,K$, we denote by $e_k$ the vector of $\mathbb{N}^K$
with $1$ in component $k$ and $0$ elsewhere.

As in Section \ref{subsec:bf}, the resources are allocated by applying balanced fairness in this capacity set
under some vector of traffic intensity $\rho$ which satisfies the stability constraints.
For simplicity of notation, for each $a \in {\cal N}$,
we denote by $\pi(a)$ the probability that the number of active queues in part $k$ is $a_k$ for each $k = 1,\ldots,K$:
$$
\pi(a)
= \mathbb{P}\{ |I({\bf X})|_\Sigma = a \}
= \sum_{\substack{A \subset I, \\ |A|_\Sigma = a}} \pi(A).
$$
For each $k = 1,\ldots,K$,
let $L_k = \mathbb{E}\left[ \sum_{i \in I_k} {\bf X}_i \right]$
denote the mean number of jobs in the queues of part $k$ and,
for each $a \in {\cal N}$,
let $L_k(a) = \mathbb{E}\left[ \sum_{i \in I_k} {\bf X}_i | |I({\bf X})|_\Sigma = a \right]$
denote the mean number of jobs in the queues of part $k$
given that there are $a_l$ active queues in part $l$ for each $l = 1,\ldots,K$.
The regularity assumptions ensure that, for each $k = 1,\ldots,K$,
$\frac1{n_k} L_k$ and $\frac1{n_k} L_k(a)$ for each $a \in {\cal N}$
also give the mean numbers of jobs at queue $i$
for any $i \in I_k$.
By the law of total expectation, we have
$$
L_k = \sum_{a \in {\cal N}} L_k(a) \pi(a),
\quad \forall k = 1,\ldots,K.
$$
The following theorem gives a recursive formula for $\pi(a)$ and $L_k(a)$
that allows one to compute recursively these quantities with a complexity $O(n_1 \cdots n_K)$.
The proof is given in Appendix \ref{app:recpoly}.

\begin{theorem}
  \label{theo:recpoly}
  Consider a system of $n$ processor-sharing queues
  with state-dependent service rates allocated according to balanced fairness
  in a polymatroid capacity set ${\cal C}$.
  Assume that ${\cal C}$ is poly-symmetric with respect to partition $\Sigma$
  and denote by $h$ the corresponding cardinality rank function.
  Further assume that for each $k = 1,\ldots,K$,
  all queues of $I_k$ receive jobs with the same traffic intensity $\varrho_k$,
  i.e.\ $\rho_i = \varrho_k$ for all $i \in I_k$.
  For each $a \in {\cal N} \setminus \{0\}$, we have
  \begin{equation}
    \label{eq:recpipoly}
    \pi(a)
    = \frac{\sum_{k=1}^K (n_k - a_k + 1) \varrho_k \pi(a-e_k)}{h(a) - \sum_{k=1}^K a_k \varrho_k}.
  \end{equation}
  Let $k = 1,\ldots,K$.
  For each $a \in {\cal N}$, we have $L_k(a) = 0$ if $a_k = 0$, and otherwise
  \begin{align}
    \label{eq:recLpoly}
    \pi(a) L_k(a)
    = \frac{
      a_k \varrho_k \pi(a) + (n_k - a_k + 1) \varrho_k \pi(a-e_k)
      + \sum_{l=1}^K (n_l - a_l + 1) \varrho_l \pi(a-e_l) L_k(a-e_l)
    }{h(a) - \sum_{l=1}^K a_l \varrho_l}.
  \end{align}
\end{theorem}

\noindent This result applies to Example \ref{ex:poly} with the partition $\Sigma = ( \{1,3\}, \{2\} )$
when classes $1$ and $3$ have the same traffic intensity.
The set of suitable vectors of traffic intensities is depicted as the darkly shaded region in Figure \ref{fig:expolysymregion}.

In this theorem, we have assumed that the cardinality rank function $h$ was given.
Given a real system like those of Sections \ref{subsec:modeltree} and \ref{subsec:modelcluster}
which is known to be poly-symmetric with regard to some partition $\Sigma = (I_k : k = 1,\ldots,K)$,
one could ask if it is also possible to build $h$ with a complexity $O(n_1 \cdots n_K)$.
This is straightforward for a computer cluster.
Concerning the tree data networks, we can actually apply
a method similar to that of the proof of Theorem \ref{theo:tree}.
Specifically, we first define recursively a concave function $f$ on ${\cal N}$ by
$f(0) = 0$,
$f(a) = C_L$ if there is $L \in {\cal T}$ so that $|L|_\Sigma = a$,
and otherwise
$$
f(a) = \min\{f(b) + f(c): b, c \in {\cal N} \text{ s.t. } b,c \neq a \text{ and } a = b + c\},
$$
from which we can construct $h$ by letting
$$
h(a) = \min\{ f(b): b \in {\cal N} \text{ and } a \le b \},
\quad \forall a \in {\cal N}.
$$

We will now see two examples of real systems
where this result applies.

\subsection{Application to tree data networks}
\label{subsec:multisource}

We consider the simple example of a tree data network
where each user has an individual access line
and all users share an aggregation link which has a capacity $C$ in bit/s.
The user access lines can have $K$ different capacities $r_1,\ldots,r_K$ in bit/s.
This corresponds to the model introduced in \cite{BC16-3} to predict
some performance metrics in Internet service provider access networks,
where the individual access lines represent subscriber lines
which are connected to the aggregation link by the digital subscriber line access multiplexer (DSLAM).

\begin{example}
  \begin{figure}[h]
    \centering
    \begin{tikzpicture}
      \node (c1) {};
      \node at ($(c1)+(0,-.6cm)$) (c2) {};
      \node at ($(c2)+(0,-.6cm)$) (c3) {};
      \node at ($(c3)+(0,-1cm)$) (c4) {};
      \node at ($(c4)+(0,-.6cm)$) (c5) {};
      \node at ($(c1)!.5!(c5)+(3.5cm,0)$) (c) {};

      \node[
        cylinder,draw=black,aspect=0.7,
        minimum height=1.2cm,minimum width=.5cm,
        cylinder uses custom fill,
        cylinder body fill=myblue!40,
        cylinder end  fill=myblue!20
      ]
      at (c1) {};
      \node[
        cylinder,draw=black,aspect=0.7,
        minimum height=1.2cm,minimum width=.5cm,
        cylinder uses custom fill,
        cylinder body fill=myblue!40,
        cylinder end  fill=myblue!20
      ]
      at (c2) {};
      \node[
        cylinder,draw=black,aspect=0.7,
        minimum height=1.2cm,minimum width=.5cm,
        cylinder uses custom fill,
        cylinder body fill=myblue!40,
        cylinder end  fill=myblue!20
      ]
      at (c3) {};
      \node[
        cylinder,draw=black,aspect=0.7,
        minimum height=1.2cm,minimum width=.5cm,
        cylinder uses custom fill,
        cylinder body fill=myblue!40,
        cylinder end  fill=myblue!20
      ]
      at (c4) {};
      \node[
        cylinder,draw=black,aspect=0.7,
        minimum height=1.2cm,minimum width=.5cm,
        cylinder uses custom fill,
        cylinder body fill=myblue!40,
        cylinder end  fill=myblue!20
      ]
      at (c5) {};
      \node[
        cylinder,draw=black,aspect=0.7,
        minimum height=1.7cm,minimum width=1cm,
        cylinder uses custom fill,
        cylinder body fill=myblue!40,
        cylinder end  fill=myblue!20
      ]
      at (c) {};

      \node at ($(c1)+(-1.3cm,0)$) {};
      \draw[-,rounded corners=.05cm]
      ($(c1)+(-1.1cm,0)$) -- ($(c1)+(-.55cm,0)$)
      ($(c1)+(.6cm,0)$) -- ($(c)+(-.8cm,.22cm)$)
      ($(c)+(.85cm,.22cm)$) -- ($(c)+(1.2cm,.22cm)$);

      \node at ($(c2)+(-1.3cm,0)$) {};
      \draw[-,rounded corners=.05cm]
      ($(c2)+(-1.1cm,0)$) -- ($(c2)+(-.55cm,0)$)
      ($(c2)+(.6cm,0)$) -- ($(c)+(-.8cm,.11cm)$)
      ($(c)+(.85cm,.11cm)$) -- ($(c)+(1.2cm,.11cm)$);

      \node at ($(c3)+(-1.3cm,0)$) {};
      \draw[-,rounded corners=.05cm]
      ($(c3)+(-1.1cm,0)$) -- ($(c3)+(-.55cm,0)$)
      ($(c3)+(.6cm,0)$) -- ($(c)+(-.8cm,0)$)
      ($(c)+(.85cm,0)$) -- ($(c)+(1.2cm,0)$);

      \node at ($(c4)+(-1.3cm,0)$) {};
      \draw[-,rounded corners=.05cm]
      ($(c4)+(-1.1cm,0)$) -- ($(c4)+(-.55cm,0)$)
      ($(c4)+(.6cm,0)$) -- ($(c)+(-.8cm,-.11cm)$)
      ($(c)+(.85cm,-.11cm)$) -- ($(c)+(1.2cm,-.11cm)$);

      \node at ($(c5)+(-1.3cm,0)$) {};
      \draw[-,rounded corners=.05cm]
      ($(c5)+(-1.1cm,0)$) -- ($(c5)+(-.55cm,0)$)
      ($(c5)+(.6cm,0)$) -- ($(c)+(-.8cm,-.22cm)$)
      ($(c)+(.85cm,-.22cm)$) -- ($(c)+(1.2cm,-.22cm)$);

      \draw [decorate,decoration={brace,amplitude=4pt}]
      ($(c3)+(-1.5cm,-.2cm)$) -- ($(c1)+(-1.5cm,.2cm)$)
      node[align=center] [black,midway,xshift=-1.6cm] {Access lines \\ with capacity $r_1$};

      \draw [decorate,decoration={brace,amplitude=4pt}]
      ($(c5)+(-1.5cm,-.2cm)$) -- ($(c4)+(-1.5cm,.2cm)$)
      node[align=center] [black,midway,xshift=-1.6cm] {Access lines \\ with capacity $r_2$};

      \node[align=center] at ($(c)+(0,1cm)$)
      {Aggregation link \\ with capacity $C$};
    \end{tikzpicture}
    \caption{User routes}
    \label{fig:multisource}
  \end{figure}
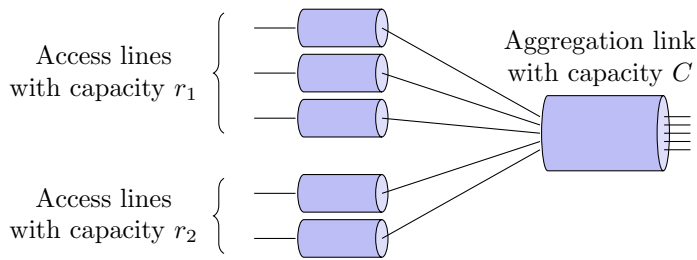
  Figure \ref{fig:multisource} gives a toy example with $K = 2$ possible access rates $r_1$ and $r_2$.
  There are three users with access rate $r_1$ and two users with access rate $r_2$.
  All users are constrained by the aggregation link with capacity $C$.
\end{example}

For each $k = 1,\ldots,K$ we denote by $I_k$ the set of users with access rate $r_k$.
These form a partition of the set $I = \bigsqcup_{k=1}^K I_k$ of users.
Theorem \ref{theo:tree} ensures that
the capacity set of this data network
is a polymatroid with rank function $\mu$ given by
\begin{equation}
  \label{eq:mums}
  \mu(A)
  = \min\left( \sum_{k=1}^K |A \cap I_k| r_k, C \right),
  \quad \forall A \subset I.
\end{equation}
It is poly-symmetric with respect to partition $\Sigma = (I_k: k = 1,\ldots,K)$.
The corresponding cardinality rank function $h$ is given by
$$
h(a) = \min\left(
  \sum_{k=1}^K a_k r_k,
  C
\right),
\quad \forall a \in {\cal N}.
$$
We further assume that for each $k = 1,\ldots,K$, all users with access line $r_k$ have the same traffic intensity $\varrho_k < r_k$.
Then the network is stable whenever $\sum_{k=1}^K n_k \varrho_k < C$,
and it meets the conditions of Theorem \ref{theo:recpoly}.

A metric of interest is the mean throughput per user.
For each $i \in I$,
we denote by $\mathbb{P}_i$ and $\mathbb{E}_i$
the conditional probability measure and expectation given that user $i$ is active,
corresponding to the stationary distribution $\pi_i(x) \propto 1_{x_i > 0} \pi(x)$.
For each $k = 1,\ldots,K$ and each $i \in I_k$,
the mean throughput perceived by user $i$ is then given by
$$
\mathbb{E}_i[\phi_i({\bf X})]
= \frac{\mathbb{E}[\phi_i({\bf X})]}{\mathbb{P}\{ {\bf X}_i > 0 \}}
= \frac{\varrho_k}{\mathbb{P}\{ {\bf X}_i > 0 \}},
$$
where the second equality holds by the conservation equation $\varrho_k = \mathbb{E}[\phi_i({\bf X})]$ for all $i \in I_k$.
Using the notations of Section \ref{subsec:recpoly},
the mean throughput of the users with access rate $r_k$ is given by
$$
\gamma_k = \frac{\varrho_k}
{1 - \sum\limits_{a \in {\cal N}: a_k < n_k} \binom{n_k-1}{a_k} \prod\limits_{l \neq k} \binom{n_l}{a_l} \pi(a)}.
$$
where $\pi(a)$ for each $a \in {\cal N}$ can be computed with a complexity $O(n_1 \cdots n_K)$ by \eqref{eq:recpipoly}.
Other performance metrics such as the mean congestion rate per user can be computed similarly.

\subsection{Application to computer clusters}
\label{subsec:polycluster}

Let $d_1,d_2 \ge 1$.
We consider a computer cluster with $m = d_1 d_2$ servers and $n = d_1 + d_2$ classes.
All servers have the same unit capacity and all jobs have a unit mean size.
The set $I$ of classes is partitioned into two subsets $I_1$ and $I_2$.
$I_1$ contains $d_2$ classes that can each be served by $d_1$ servers
and $I_2$ contains $d_1$ classes that can each be served by $d_2$ servers.
For any $i = 1,\ldots,d_2$, the $i$-th class of $I_1$
can be served by the servers $(i-1) d_1 + j$ for $j = 1,\ldots,d_1$.
For any $i = 1,\ldots,d_1$, the $i$-th class of $I_2$
can be served by the servers $i + (j-1) d_1$ for $j = 1,\ldots,d_2$.
Figure \ref{fig:polycluster} gives a toy example with $d_1 = 2$ and $d_2 = 3$.

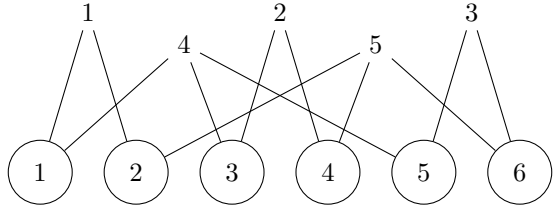
\begin{figure}[h]
  \centering
  \begin{tikzpicture}
    \def\scale{.85}

    \node[draw,circle,minimum size=\scale*1cm] (server1) {1};
    \node[draw,circle,minimum size=\scale*1cm]
    at ($(server1)+(\scale*1.5cm,0)$) (server2) {2};
    \node[draw,circle,minimum size=\scale*1cm]
    at ($(server2)+(\scale*1.5cm,0)$) (server3) {3};
    \node[draw,circle,minimum size=\scale*1cm]
    at ($(server3)+(\scale*1.5cm,0)$) (server4) {4};
    \node[draw,circle,minimum size=\scale*1cm]
    at ($(server4)+(\scale*1.5cm,0)$) (server5) {5};
    \node[draw,circle,minimum size=\scale*1cm]
    at ($(server5)+(\scale*1.5cm,0)$) (server6) {6};

    \node at ($(server1)!.5!(server2)+(0,\scale*2.5cm)$) (class1) {1};
    \node at ($(server3)!.5!(server4)+(0,\scale*2.5cm)$) (class2) {2};
    \node at ($(server5)!.5!(server6)+(0,\scale*2.5cm)$) (class3) {3};
    \node at ($(server2)!.5!(server3)+(0,\scale*2cm)$) (class4) {4};
    \node at ($(server4)!.5!(server5)+(0,\scale*2cm)$) (class5) {5};

    \draw[-] (class1) -- (server1);
    \draw[-] (class1) -- (server2);
    \draw[-] (class2) -- (server3);
    \draw[-] (class2) -- (server4);
    \draw[-] (class3) -- (server5);
    \draw[-] (class3) -- (server6);
    \draw[-] (class4) -- (server1);
    \draw[-] (class4) -- (server3);
    \draw[-] (class4) -- (server5);
    \draw[-] (class5) -- (server2);
    \draw[-] (class5) -- (server4);
    \draw[-] (class5) -- (server6);
  \end{tikzpicture}
  \caption{Computer cluster with $d_1 = 2$ and $d_2 = 3$}
  \label{fig:polycluster}
\end{figure}

Any class of $I_1$ shares exactly one server with any class of $I_2$,
and this server is dedicated to these two classes.
The rank function of this cluster is thus given by
$$
\mu(A) = |A \cap I_1| d_1 + |A \cap I_2| d_2 - |A \cap I_1| \times |A \cap I_2|,
\quad \forall A \subset I.
$$
The polymatroid capacity set defined by this rank function is poly-symmetric with respect to partition $\Sigma = (I_1, I_2)$
and the corresponding cardinality rank function is given by
$$
h(a) = a_1 d_1 + a_2 d_2 - a_1 a_2,
\quad \forall a \in {\cal N}.
$$

For each $k = 1,2$, assume that all classes in $I_k$ have the same traffic intensity $\varrho_k$.
Further assume that the vector of traffic intensities $\varrho = (\varrho_1, \varrho_2)$ stabilizes the system, that is
$$
a_1 \varrho_1 + a_2 \varrho_2 < a_1 d_1 + a_2 d_2 - a_1 a_2,
\quad \forall a \in {\cal N}.
$$
We can then apply Theorem \ref{theo:recpoly} with partition $\Sigma$
to compute the mean number of jobs of each class with a complexity $O(n_1 n_2)$.
We deduce the mean delay $\delta_i$ of class-$i$ jobs
for each $i \in I$ by Little's law:
$$
\delta_i = \frac{L_k}{n_k \lambda_i},
\quad \forall k = 1,2,
\quad \forall i \in I_k.
$$

\section{Stochastic Bounds}
\label{sec:bounds}

\subsection{Monotonicity result}
\label{subsec:monotonicity}

While the property of poly-symmetry is not often satisfied in practice,
except in specific cases like the examples of Sections 3.3 and 3.4,
it can be used to derive stochastic bounds on most systems, as shown below.
The following result will allow us to control the impact of the capacity set on performance.

Given $0 < \epsilon < 1$ and a polymatroid ${\cal C}$ in $\mathbb{R}_+^n$ with rank function $\mu$,
we denote by $(1+\epsilon) {\cal C}$ the polymatroid in $\mathbb{R}_+^n$ with rank function $(1+\epsilon) \mu$
and by $(1-\epsilon) {\cal C}$ the polymatroid in $\mathbb{R}_+^n$ with rank function $(1-\epsilon) \mu$.

\begin{theorem}
  \label{theo:mono}
  Let $0 < \epsilon < 1$.
  Consider two polymatroids $\hat{\cal C}$ and ${\cal C}$ in $\mathbb{R}_+^n$
  such that $\hat{\cal C}$ is a subset of $(1+\epsilon) {\cal C}$
  and a superset of $(1-\epsilon) {\cal C}$.
  Let $\rho$ be an element in the interior of $(1-\epsilon) {\cal C}$
  and denote respectively by $\pi$, $\pi_+$ and $\pi_-$ the steady state distributions
  of the processor-sharing systems with capacity sets
  $\hat{\cal C}$, $(1 + \epsilon) {\cal C}$ and $(1 - \epsilon) {\cal C}$
  under traffic intensity $\rho$.
  Then
  $$
  \frac{\pi_-(0)}{\pi_+(0)} \pi_+(x)
  \le \pi(x)
  \le \frac{\pi_+(0)}{\pi_-(0)} \pi_-(x),
  \quad \forall x \in \mathbb{N}^n.
  $$
  Specifically, for each $i \in I$, we have
  $$
  \frac{\pi_-(0)}{\pi_+(0)} L_{i,+}
  \le L_i
  \le \frac{\pi_+(0)}{\pi_-(0)} L_{i,-},
  $$
  where $L_i$, $L_{i,+}$ and $L_{i,-}$ are the mean number of job at queue $i$
  under distributions $\pi$, $\pi_+$ and $\pi_-$ respectively.
\end{theorem}

\begin{proof}
  Denote by $\hat\mu$ and $\mu$ the rank functions of $\hat{\cal C}$ and ${\cal C}$ respectively.
  Let $\Phi$, $\Phi_+$ and $\Phi_-$ denote the balance functions of the resource allocations
  defined by balanced fairness in the capacity sets
  $\hat{\cal C}$, $(1 + \epsilon) {\cal C}$ and $(1 - \epsilon) {\cal C}$ respectively.
  We first prove by induction on $|x|$ that
  $$
  \Phi_+(x)
  \le \Phi(x)
  \le \Phi_-(x),
  \quad \forall x \in \mathbb{N}^n.
  $$
  The property holds for $x = 0$.
  Let $x \in \mathbb{N}^n \setminus \{0\}$
  and assume the inequality is valid for any $y \in \mathbb{N}^n$
  such that $|y| < |x|$.
  Then we have by \eqref{eq:recPhi}:
  $$
  \Phi(x)
  = \frac{\sum_{i \in I(x)} \Phi(x-e_i)}{\hat\mu(I(x))}
  \le \frac{\sum_{i \in I(x)} \Phi_-(x-e_i)}{\hat\mu(I(x))}
  \le \frac{\sum_{i \in I(x)} \Phi_-(x-e_i)}{(1 - \epsilon) \mu(I(x))}
  = \Phi_-(x)
  $$
  where the first inequality holds by the induction assumption
  and the second holds by the inclusion of $(1-\epsilon) {\cal C}$ into $\hat{\cal C}$.
  We prove the other side of the inequality
  by using the inclusion of $\hat{\cal C}$ into $(1+\epsilon) {\cal C}$.
  This completes the proof by induction.

  It follows that
  $$
  \frac1{\pi(0)}
  = \sum_{x \in \mathbb{N}^n} \Phi(x) \rho^x
  \ge \sum_{x \in \mathbb{N}^n} \Phi_+(x) \rho^x
  = \frac1{\pi_+(0)}.
  $$
  Thus for each $x \in \mathbb{N}^n$, we obtain
  $$
  \pi(x)
  = \pi(0) \Phi(x) \rho^x
  \le \pi(0) \Phi_-(x) \rho^x
  \le \pi_+(0) \Phi_-(x) \rho^x
  = \frac{\pi_+(0)}{\pi_-(0)} \pi_-(x).
  $$
  The proof for the other part of the inequality is similar.
  The second inequality about the mean number of jobs follows by summation.
\end{proof}

The following sections illustrate how we can apply this result to the models of
tree data networks and computer clusters.

\subsection{Application to tree data networks}

We first use this result to relax some assumptions of Section \ref{subsec:multisource}.
The flows of each user go through
an individual access line which is dedicated to this user
and an aggregation link shared by all users.
We still consider $K$ groups $I_1,\ldots,I_K$ of users which form a partition of the set $I$.
For each $k = 1,\ldots,K$, the access rates of the users in $I_k$ may be different
but we assume that they are all between $(1 - \epsilon) r_k$ and $(1 + \epsilon) r_k$ for some $\epsilon > 0$.
Similarly, the capacity of the aggregation link is between $(1 - \epsilon) C$ and $(1 + \epsilon) C$.
The corresponding polymatroid capacity set $\hat{\cal C}$ is not poly-symmetric
with respect to partition $\Sigma = (I_k : k = 1,\ldots,K)$ any more
but its rank function $\hat\mu$ satisfies:
$$
(1 - \epsilon) \mu(A)
\le \hat\mu(A)
\le (1 + \epsilon) \mu(A),
\quad \forall A \subset I,
$$
where $\mu$ is the rank function defined by \eqref{eq:mums}.
Denoting by ${\cal C}$ the polymatroid defined by $\mu$,
it follows that $\hat{\cal C}$ is a superset of $(1 - \epsilon) {\cal C}$ and a subset of $(1 + \epsilon) {\cal C}$.
We can thus apply Theorem \ref{theo:mono}.
In the special case where for each $k = 1,\ldots,K$,
all users of $I_k$ have the same traffic intensity $\varrho_k < (1 - \epsilon) r_k$,
with $\sum_{k=1}^K n_k \varrho_k < (1 - \epsilon) C$,
we can use Theorem \ref{theo:recpoly} to compute the bounds.

Specifically, let $\pi_+$ and $\pi_-$ denote the steady state distributions
of the processor-sharing systems with capacity sets $(1 + \epsilon) {\cal C}$ and $(1 - \epsilon) {\cal C}$ respectively
under traffic intensity $\rho$.
For each $k = 1,\ldots,K$ and each $i \in I_k$,
the mean throughput $\gamma_i$ of user $i$
satisfies
$$
\frac{\pi_-(0)}{\pi_+(0)} \gamma_{k,-}
\le \gamma_i
\le \frac{\pi_+(0)}{\pi_-(0)} \gamma_{k,+},
$$
where $\gamma_{k,+}$ and $\gamma_{k,-}$ are the mean throughputs
under distributions $\pi_+$ and $\pi_-$ respectively.
We have
$$
\gamma_{k,\pm} = \frac{\varrho_k}
{1 - \sum\limits_{a \in {\cal N}: a_k < n_k} \binom{n_k-1}{a_k} \prod\limits_{l \neq k} \binom{n_l}{a_l} \pi_\pm(a)},
\quad \forall k = 1,\ldots,K,
$$
and by \eqref{eq:recpipoly}
$$
\pi_\pm(a)
= \frac{\sum_{k=1}^K (n_k - a_k + 1) \varrho_k \pi_\pm(a-e_k)}{(1 \pm \epsilon) h(a) - \sum_{k=1}^K a_k \varrho_k},
\quad \forall a \in {\cal N} \setminus \{0\}.
$$

\subsection{Application to computer cluster with random assignment}
\label{subsec:randomassign}

\paragraph{Random assignment.}
Consider a cluster as described in Section \ref{subsec:modelcluster},
where we denote by $S = \{1,\ldots,m\}$ the set of servers and by $I$ the set of class indices.
Let $K \ge 1$ and consider for simplicity a partition $\Sigma = (I_1,\ldots,I_K)$ of $I$ into $K$ parts of size $n$,
so that the total number of job classes in the cluster is now given by $Kn$.
We can easily generalize the result to $K$ parts of different sizes.
We use the same notation as in Sections \ref{subsec:defpoly} and \ref{subsec:recpoly}:
for each $A \subset I$,
$a = |A|_\Sigma$ denotes the $K$-dimensional vector
whose $k$-th component is $a_k = |A \cap I_k|$, the size of the $k$-th part of $A$ in partition $\Sigma$,
for each $k = 1,\ldots,K$;
the set of these vectors is denoted by ${\cal N} = \{0,1,\ldots,n\}^K$.

We now introduce a random assignment of the servers to the job classes,
which is described by an assignment graph with random edges.
Each realization of this random assignment defines a polymatroid capacity set as in Section \ref{subsec:capacityset}.
Hence, once this initial assignment is settled, we can apply balanced fairness over the associated capacity set as described in Section \ref{subsec:bf}.
Note that our assignment is static in the sense that the assignment graph does not change with time.
For a given realization, we can thus observe the evolution of the cluster under stochastic arrivals
and compute the resulting performance metrics as we did in Section \ref{subsec:perf}.

The servers are randomly assigned to the job classes as follows.
Let $d = (d_k : k = 1,\ldots,K)$
be a vector of positive integers.
For any $k = 1,\ldots,K$ and $i \in I_k$,
the set ${\bf S}_i$ of servers that can process class-$i$ jobs
is chosen uniformly and independently at random
among the subsets of $S = \{1,\ldots,m\}$ of cardinality $d_k$.
As in Section \ref{subsec:modelcluster},
the random assignment is described
by the family $({\bf S}_i : i \in I)$
which defines a random bipartite graph
$$
{\bf G} = \left( I, S, \bigcup_{i \in I} ( \{i\} \times {\bf S}_i ) \right)
$$
with deterministic sets of vertices $I$ and $S$
and a random set of edges.
Each realization
$(S_i : i \in I)$ of the random assignment
defines a polymatroid capacity set
with a rank function given by \eqref{eq:mucluster}.
This allows us to define a random rank function associated with the random assignment by
$$
{\bf M}(A)
= \sum_{s \in S} \mu_s 1_{s \in \bigcup_{i \in A} {\bf S}_i},
\quad \forall A \subset I.
$$
Now let $\mu$ denote the corresponding mean rank function:
$$
\mu(A)
= \mathbb{E}[{\bf M}(A)],
\quad \forall A \subset I.
$$
The following lemma proves that the polymatroid
defined by $\mu$ is poly-symmetric with respect to $\Sigma$.

\begin{lemma}
  \label{lem:meanrank}
  For each $A \subset I$, we have
  $\mu(A) = \xi m p_a$
  with
  $\xi = \frac1m \sum_{s \in S} \mu_s$,
  $a = |A|_\Sigma$
  and
  $$
  p_a = 1 - \prod_{k=1}^K \left( 1 - \frac{d_k}m \right)^{a_k}.
  $$
\end{lemma}

\begin{proof}
  Let $a \in {\cal N}$ and consider any set $A \subset I$
  with $|A|_\Sigma = a$.
  We just need to observe that
  $$
  \mu(A)
  = \sum_{s \in S} \mu_s
  \mathbb{P} \left\{ s \in \bigcup_{i \in A} {\bf S}_i \right\}.
  $$
  For each $k = 1,\ldots,K$ such that $a_k > 0$,
  the probability that a server can not serve a specific class of $A \cap I_k$
  is $\binom{m-1}{d_k} / \binom{m}{d_k} = 1 - \frac{d_k}m$.
  Since the assignments of the classes are independent,
  it follows that the probability that this server can serve at least one class in $A$ is given by $p_a$.
\end{proof}

Let $\rho = (\rho_i : i \in I)$ be a vector of traffic intensities.
If $G$ is a realization of the random assignment graph ${\bf G}$
such that $\rho$ is in the interior of the polymatroid capacity set defined by $G$,
then the corresponding processor-sharing system is stable under balanced fairness
and we can study its steady-state behavior.
We thus denote by ${\bf X} = ({\bf X}_i : i \in I)$ the random vector distributed according to the stationary distribution of the system state
when $\rho$ is in the interior of ${\bf G}$,
and for completeness we let ${\bf X} = 0$ otherwise.
For each realization $G$ of the random assignment such that $\rho$ is in the interior of the corresponding capacity set, we let
$$
L_i(G) = \mathbb{E}\left[ {\bf X}_i | {\bf G} = G \right],
\quad \forall i \in I,
$$
which is simply the mean number of jobs of each class in the corresponding processor-sharing system under balanced fairness, as defined in Section \ref{subsec:perf}.
For each realization $G$ which is not stabilized by $\rho$, we let $L_i(G) = +\infty$.

\paragraph{Asymptotic poly-symmetry.}

We consider a sequence of computer clusters with random assignment as defined in the previous section.
Let $K \ge 1$ and $b > 0$.

For each $n \ge 1$,
the $n$-th random cluster of the sequence contains
$m^{(n)} = \lceil bn \rceil$ servers and $Kn$ job classes.
We denote the set of servers by $S^{(n)} = \{1,\ldots,m^{(n)}\}$ and
the set of job classes by $I^{(n)} = \{1,\ldots,Kn\}$.
The service rate of server $s$ is $\mu_s^{(n)}$ for each $s \in S^{(n)}$.
For simplicity, we consider a partition
$\Sigma^{(n)} = (I_k^{(n)} : k = 1,\ldots,K)$
of $I^{(n)}$ into $K$ parts of size $n$.
We can generalize the result to $K$ parts of different sizes
as long as the size of each part is linear in $n$.
For each $k = 1,\ldots,K$ and $i \in I_k^{(n)}$,
the set ${\bf S}_i^{(n)}$ of servers that can process class-$i$ jobs
is chosen uniformly and independently at random
among the subsets of $S^{(n)}$ of cardinality $d_k^{(n)}$.
Let ${\bf G}^{(n)}$ denote the random graph defined by this random assignment,
${\bf M}^{(n)}$ the corresponding random rank function
and $\mu^{(n)}$ its expectation.

By Lemma \ref{lem:meanrank},
for any $n \ge 1$ and $a \in {\cal N}^{(n)} = \{0,\ldots,n\}^K$,
we have
$$
\mu^{(n)}(A)
= \mathbb{E}\left[ {\bf M}^{(n)}(A) \right]
= \xi^{(n)} m^{(n)} p_a^{(n)}
\quad \text{with} \quad
p_a^{(n)} = 1 - \prod_{k=1}^K \left( 1 - \frac{d_k^{(n)}}m \right)^{a_k}
$$
for all set $A \subset I^{(n)}$ with $|A|_{\Sigma^{(n)}} = a$,
where $\xi^{(n)} = \frac1{m^{(n)}} \sum_{s \in S^{(n)}} \mu_s^{(n)}$ is the mean server capacity.
Theorem \ref{theo:concentration} below shows that,
under the following two assumptions on the server capacities and the degrees of parallelism,
the probability that the random rank function is uniformly close to its mean
is $1 - o\left( \frac1n \right)$.
The proof is given in Appendix \ref{app:concentration}.

\begin{ass}
  \label{ass:servers}
  For each $n \ge 1$,
  $S^{(n)}$ is partitioned into a constant number of groups.
  Each group contains $\Omega(n)$ servers which have the same capacity.
\end{ass}

\begin{ass}
  \label{ass:degrees}
  For each $k = 1,\ldots,K$, the sequence $\left( d_k^{(n)} : n \ge 1 \right)$ satisfies
  $d_k^{(n)} = \omega(\log n)$.
\end{ass}

\begin{theorem}
  \label{theo:concentration}
  Let $0 < \epsilon < 1$.
  Under Assumptions \ref{ass:servers} and \ref{ass:degrees},
  there exists a sequence $(g_n : n \ge 1)$ such that
  $g_n = \omega(\log n)$ and for any $n \ge 1$,
  $$
  \mathbb{P} \left\{ \exists A \subset I^{(n)}
    \text{ s.t. } {\bf M}^{(n)}(A) \le (1 - \epsilon) \mu^{(n)}(A)
  \right\}
  \le e^{-g_n}
  $$
  and
  $$
  \mathbb{P} \left\{ \exists A \subset I^{(n)}
    \text{ s.t. } {\bf M}^{(n)}(A) \ge (1 + \epsilon) \mu^{(n)}(A)
  \right\}
  \le e^{-g_n}.
  $$
\end{theorem}

\noindent
Corollary \ref{coro:concentration}
follows from Theorem \ref{theo:concentration}.
For any $n \ge 1$, let ${\cal C}^{(n)}$
denote the polymatroid defined by the rank function $\mu^{(n)}$.
${\cal C}^{(n)}$ is poly-symmetric with respect to the partition $\Sigma^{(n)}$.

\begin{coro}
  \label{coro:concentration}
  Let $0 < \epsilon < 1$.
  Under Assumptions \ref{ass:servers} and \ref{ass:degrees},
  the random capacity set resulting from the random assignment
  is a subset of $(1+\epsilon) {\cal C}^{(n)}$
  and a superset of $(1-\epsilon) {\cal C}^{(n)}$
  with probability $1 - o\left( \frac1n \right)$.
\end{coro}

\paragraph{Performance metrics.}

For each $n \ge 1$, we consider a vector $\varrho^{(n)} \in \mathbb{R}_+^K$ of traffic intensities per part
which stabilizes the processor-sharing system with $n$ queues and capacity set $(1-\epsilon) {\cal C}^{(n)}$ under balanced fairness, that is,
$$
\sum_{k=1}^K a_k \varrho_k^{(n)}
< (1 - \epsilon) h^{(n)} (a),
\quad \forall a \in {\cal N}^{(n)},
$$
where $h^{(n)}$ is the cardinality rank function of the mean capacity set with respect to partition $\Sigma^{(n)}$:
$$
h^{(n)} (a)
= \xi^{(n)} m^{(n)} p_a^{(n)},
\quad \forall a \in {\cal N}^{(n)}.
$$
Let ${\bf X}^{(n)}$ denote
the state in the $n$-th randomized computer cluster
when we allocate the resources according to balanced fairness.
Given a realization $G^{(n)}$ which
is stabilized by $\varrho^{(n)}$,
the mean numbers of jobs per class are given by
$$
L_i\left( G^{(n)} \right)
= \mathbb{E} \left[
  {\bf X}_i^{(n)} \big| {\bf G}^{(n)} = G^{(n)}
\right],
\quad \forall i \in I^{(n)}.
$$
For each realization $G^{(n)}$ which is not stabilized by $\varrho^{(n)}$,
we let $L_i\left( G^{(n)} \right) = +\infty$.
This allows us to define the random variables
$$
{\bf L}_i^{(n)}
= L_i\left( {\bf G}^{(n)} \right),
\quad \forall i \in I.
$$

Combining Theorems \ref{theo:mono} and \ref{theo:concentration} yields the following result.

\begin{theorem}
  \label{theo:boundcluster}
  Let $0 < \epsilon < 1$.
  For any $n \ge 1$,
  denote by $\pi_+^{(n)}$ and $\pi_-^{(n)}$ the stationary distributions of
  the processor-sharing systems with $n$ queues and capacity sets $(1 + \epsilon) {\cal C}^{(n)}$ and $(1 - \epsilon) {\cal C}^{(n)}$ respectively,
  when the traffic intensity of the classes in $I_k^{(n)}$ is $\varrho_k^{(n)}$, for any $k = 1,\ldots,K$.
  Let $L_{k,+}^{(n)}$ and $L_{k,-}^{(n)}$ denote the corresponding mean number of jobs per queue in part $k$, for each $k = 1,\ldots,K$.

  Under Assumptions \ref{ass:servers} and \ref{ass:degrees}, we have
  $$
  \mathbb{P} \left\{
    \frac{\pi_-^{(n)}(0)}{\pi_+^{(n)}(0)} \frac{L_{k,+}^{(n)}}n
    \le {\bf L}_i^{(n)}
    \le \frac{\pi_+^{(n)}(0)}{\pi_-^{(n)}(0)} \frac{L_{k,-}^{(n)}}n,
    \quad \forall k = 1,\ldots,K,
    \quad \forall i \in I_k
  \right\}
  = 1 - o\left( \frac1n \right).
  $$
\end{theorem}

For each $n \ge 1$,
Theorem \ref{theo:recpoly} gives formulas to compute
$\pi_\pm^{(n)}$ and $L_{k,\pm}^{(n)}$ for each $k = 1,\ldots,K$
with a complexity $O\left( n^K \right)$.
Using Little's law, we can deduce bounds on the mean delay per class.

\begin{figure}
  \centering
  \includegraphics[scale=1]{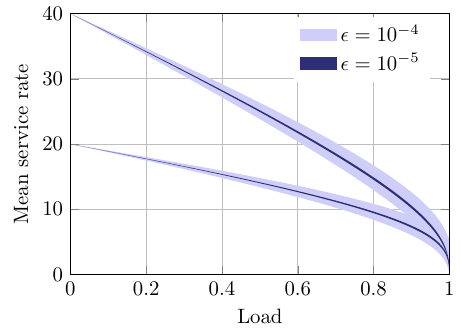}
  \caption{Bounds obtained with $m = 10,000$, $n = 1,000$, $d_1 = 20$, $d_2 = 40$}
  \label{fig:boundscluster}
\end{figure}

\begin{num}
  We omit writing the exponents for brevity.
  Consider a randomized cluster with $m = 10,000$ servers with unit capacity.
  The set of classes is partitioned into two parts $I_1$ and $I_2$ with $n = 1,000$ classes each.
  The classes of $I_1$ have a degree $d_1 = 20$ and the ones of $I_2$ have a degree $d_2 = 40$.
  All jobs have a mean size $1$ and the arrival rates are proportional to the degrees.
  The traffic intensity of any class in $I_2$ is thus twice that of any class in $I_1$.
  Let $\varrho \in \mathbb{R}_+^2$ so that
  $$
  \varrho_1 = \frac{d_1}{d_1 + d_2} \frac{h(n, n)}n
  \quad \text{and} \quad
  \varrho_2 = \frac{d_2}{d_1 + d_2} \frac{h(n, n)}n,
  $$
  where $h$ is the cardinality rank function of the mean capacity set with respect to partition $\Sigma = (I_1, I_2)$:
  $$
  h(a)
  = m \left\{
    1 - \left( 1 - \frac{d_1}m \right)^{a_1} \left( 1 - \frac{d_2}m \right)^{a_2}
  \right\},
  \quad \forall a \in \{0,1,\ldots,n\}^2.
  $$
  We can prove that the vector of traffic intensities $\rho \in \mathbb{R}_+^{2n}$
  with $\rho_i = \varrho_1$ for all $i \in I_1$
  and $\rho_i = \varrho_2$ for all $i \in I_2$
  is on the boundary of the mean capacity set.

  Let $0 < \epsilon < 1$.
  For each $\alpha \in (0,1 - \epsilon)$, $\alpha \varrho$ stabilizes
  the processor-sharing systems with capacity sets $(1 + \epsilon) {\cal C}$ and $(1 - \epsilon) {\cal C}$.
  The bounds on the mean delay that follow from Theorem \ref{theo:boundcluster} by applying Little's law are given by
  $$
  \frac{1 + \epsilon}\alpha \frac{\pi_-(0)}{\pi_+(0)} \frac{L_{k,+}}{n \varrho_k}
  \quad \text{and} \quad
  \frac{1 - \epsilon}\alpha \frac{\pi_+(0)}{\pi_-(0)} \frac{L_{k,-}}{n \varrho_k},
  $$
  where $\pi_\pm(0)$ and $L_{k,\pm}$ are computed
  with the recursion expressions of Theorem \ref{theo:recpoly} as follows.
  For each $a \in \{0,1,\ldots,n\}^2$,
  $$
  \pi_\pm(a)
  = \frac
  {(n - a_1 + 1) \varrho_1 \pi_\pm(a - e_1) + (n - a_2 + 1) \varrho_2 \pi_\pm(a - e_2)}
  {\frac{1 \pm \epsilon}\alpha h(a) - a_1 \varrho_1 - a_2 \varrho_2}.
  $$
  For any $k = 1,2$ and $a \in {\cal N}$, we have $L_{k,\pm}(a) = 0$ if $a_k = 0$, and otherwise
  \begin{multline*}
    \pi_\pm (a) L_{k,\pm} (a)
    = \frac1{\frac{1 \pm \epsilon}\alpha h(a) - a_1 \varrho_1 - a_2 \varrho_2}
    \Bigg\{
      a_k \varrho_k \pi_\pm(a) + (n - a_k + 1) \varrho_k \pi_\pm(a-e_k) \\
      + \sum_{l=1}^2 (n - a_l + 1) \varrho_l \pi_\pm(a-e_l) L_{k,\pm}(a-e_l)
    \Bigg\}.
  \end{multline*}
  Figure \ref{fig:boundscluster} gives the bounds obtained as a function of $\alpha \in (0, 1 - \epsilon)$,
  for different values of $\epsilon$.
  For simplicity, we draw the mean service rate per job,
  which is simply the inverse of the mean delay since all jobs have a unit mean size.
\end{num}

\section{Conclusion}
\label{sec:ccl}

In this paper, we consider processor-sharing systems and
introduce a poly-symmetry criterion
on the structure of their capacity set
which ensures that the performance metrics can be computed
with a complexity which is polynomial in the number of queues
if the traffic intensities per queue are adjusted accordingly.
We showed that these formulas can also be used
to bound the performance of a system when its capacity set
is nearly poly-symmetric.
We applied these results to
tree data networks and computer clusters. 

In future works, we would like to generalize the notion of poly-symmetry
to relax some symmetry assumptions imposed in this paper
while keeping a reasonable time complexity for the calculation of the performance metrics.
We also believe there is further scope of enhancing the stochastic bounds by
expanding their scope as well as obtaining tighter bounds in some specific scaling regimes.

\appendix
\normalsize
\section*{Appendix}

\section{Proof of Theorem \ref{theo:recpoly}}
\label{app:recpoly}

\subsection{Recursion \eqref{eq:recpipoly}}

Let $a \in {\cal N} \setminus \{0\}$.
By \eqref{eq:recpisets}, we have
\begin{align*}
  \pi(a)
  = \sum_{\substack{A \subset I \\ |A|_\Sigma = a}} \pi(A)
  = \sum_{\substack{A \subset I \\ |A|_\Sigma = a}}
  \frac{\sum_{i \in A} \rho_i \pi(A \setminus \{i\})}{\mu(A) - \sum_{i \in A} \rho_i}.
\end{align*}
The regularity assumptions ensure that
$\mu(A) - \sum_{i \in A} \rho_i = h(a) - \sum_{k=1}^K a_k \varrho_k$
for any $A \subset I$ with $|A|_\Sigma = a$.
Thus we obtain
\begin{align*}
  \left( h(a) - \sum_{k=1}^K a_k \varrho_k \right) \pi(a)
  = \sum_{\substack{A \subset I \\ |A|_\Sigma = a}} \sum_{i \in A} \rho_i \pi(A \setminus \{i\})
  = \sum_{k=1}^K \varrho_k \sum_{i \in I_k}
  \sum_{\substack{A \subset I, i \in A \\ |A|_\Sigma = a}} \pi(A \setminus \{i\}).
\end{align*}
For any $k = 1,\ldots,K$ and any $i \in I_k$, we do the substitution
$$
\sum_{\substack{A \subset I, i \in A \\ |A|_\Sigma = a}} \pi(A \setminus \{i\})
= \sum_{\substack{B \subset I \setminus \{i\} \\ |B|_\Sigma = a-e_k}} \pi(B),
$$
and thus we obtain for any $k = 1,\ldots,K$,
\begin{align}
  \label{eq:firsterm}
  \sum_{i \in I_k} \sum_{\substack{A \subset I, i \in A \\ |A|_\Sigma = a}} \pi(A \setminus \{i\})
  = \sum_{\substack{B \subset I \\ |B|_\Sigma = a-e_k}}
  \sum_{i \in I_k \setminus (B \cap I_k)} \pi(B)
  = (n_k - a_k + 1) \pi({a-e_k}).
\end{align}
This proves \eqref{eq:recpipoly}.

\subsection{Recursion \eqref{eq:recLpoly}}

Let $k = 1,\ldots,K$ and $a \in {\cal N} \setminus \{0\}$.
By definition of $L_k(a)$, we have
$$
L_k(a)
= \mathbb{E}\left[ \sum_{i \in I_k} {\bf X}_i \bigg| |I({\bf X})|_\Sigma = a \right]
= \sum_{i \in I_k} \mathbb{E}\left[ {\bf X}_i | |I({\bf X})|_\Sigma = a \right].
$$
It follows that
\begin{align}
  \label{eq:splitL}
  \pi(a) L_k(a)
  &= \sum_{i \in I_k} \sum_{\substack{A \subset I, i \in A \\ |A|_\Sigma = a}}
  \pi(A) L_i(A),
\end{align}
and by \eqref{eq:recLsets}, we obtain
\begin{align*}
  \pi(a) L_k(a)
  &= \sum_{i \in I_k} \sum_{\substack{A \subset I, i \in A \\ |A|_\Sigma = a}}
  \frac{\rho_i \pi(A \setminus \{i\}) + \rho_i \pi(A)
  + \sum_{j \in A \setminus \{i\}} \rho_j \pi(A \setminus \{j\}) L_i(A \setminus \{j\})}
  {\mu(A) - \sum_{j \in A} \rho_j}.
\end{align*}
Using the regularity assumptions, this can be rewritten as
\begin{align*}
  \left( h(a) - \sum_{l=1}^K a_l \varrho_l \right) \pi(a) L_k(a)
  {}={} &\varrho_k \sum_{i \in I_k} \sum_{\substack{A \subset I, i \in A \\ |A|_\Sigma = a}} \pi(A \setminus \{i\})
  + \varrho_k \sum_{i \in I_k} \sum_{\substack{A \subset I, i \in A \\ |A|_\Sigma = a}} \pi(A) \\
  &{}+{} \sum_{i \in I_k} \sum_{\substack{A \subset I, i \in A \\ |A|_\Sigma = a}}
  \sum_{\substack{j \in A \\ j \neq i}} \rho_j \pi(A \setminus \{j\}) L_i(A \setminus \{j\}).
\end{align*}
The first term is given by \eqref{eq:firsterm}.
The second term is simply
\begin{align*}
  \varrho_k \sum_{i \in I_k} \sum_{\substack{A \subset I, i \in A \\ |A|_\Sigma = a}} \pi(A)
  = \varrho_k \sum_{\substack{A \subset I \\ |A|_\Sigma = a}} \sum_{i \in A \cap I_k} \pi(A)
  = a_k \varrho_k \pi(a).
\end{align*}
Finally, for any $i \in I_k$, we have
\begin{align*}
  \sum_{\substack{A \subset I, i \in A \\ |A|_\Sigma = a}}
  \sum_{\substack{j \in A \\ j \neq i}} \rho_j \pi(A \setminus \{j\}) L_i(A \setminus \{j\})
  = \sum_{l=1}^K \varrho_l \sum_{\substack{j \in I_l \\ j \neq i}}
  \sum_{\substack{A \subset I, i,j \in A \\ |A|_\Sigma = a}}
  \pi(A \setminus \{j\}) L_i(A \setminus \{j\}).
\end{align*}
Doing the same substitution as in \eqref{eq:firsterm}, we have for any $l = 1,\ldots,K$,
\begin{align*}
  \sum_{\substack{j \in I_l \\ j \neq i}}
  \sum_{\substack{A \subset I, i,j \in A \\ |A|_\Sigma = a}}
  \pi(A \setminus \{j\}) L_i(A \setminus \{j\})
  &= \sum_{\substack{B \subset I, i \in B \\ |B|_\Sigma = a-e_l}}
  \sum_{j \in I_l \setminus (B \cap I_l)}
  \pi(B) L_i(B), \\
  &= (n_l - a_l + 1)
  \sum_{\substack{B \subset I, i \in B \\ |B|_\Sigma = a-e_l}}
  \pi(B) L_i(B).
\end{align*}
Hence the third term of the sum is equal to
\begin{align*}
  \sum_{i \in I_k} \sum_{l=1}^K (n_l - a_l + 1) \varrho_l
  \sum_{\substack{B \subset I, i \in B \\ |B|_\Sigma = a-e_l}}
  \pi(B) L_i(B)
  &= \sum_{l=1}^K (n_l - a_l + 1) \varrho_l
  \sum_{i \in I_k} \sum_{\substack{B \subset I, i \in B \\ |B|_\Sigma = a-e_l}}
  \pi(B) L_i(B), \\
  &= \sum_{l=1}^K (n_l - a_l + 1) \varrho_l \pi({a-e_l}) L_k(a-e_l),
\end{align*}
where the second equality holds by \eqref{eq:splitL}.
When we substitute the three terms by their expressions, we obtain \eqref{eq:recLpoly}.

\section{Proof of Theorem \ref{theo:concentration}}
\label{app:concentration}

We give the proof only for the case $K = 2$ for ease of notation;
the other cases follow analogously.
For now, we assume that for all $n \ge 1$,
all servers have the same capacity
$\mu_s^{(n)} = \xi^{(n)}$ for any $s \in S^{(n)}$.

Let $0 < \epsilon < 1$.
We will show that there exists a sequence $(g_n : n \ge 1)$
such that $g_n = \omega(\log n)$ and for any $n \ge 1$,
$$
\mathbb{P} \left\{ \exists A \subset I^{(n)} \text{ s.t. }
{\bf M}^{(n)}(A) \le (1 - \epsilon) \mu^{(n)}(A) \right\} \le e^{-g_n}.
$$
Let us first give the main ideas of the proof.
As in \cite{SV16}, it is divided in three steps.
We first provide a bound for
$\mathbb{P} \left\{ {\bf M}^{(n)}(A) \le (1 - \epsilon) \mu^{(n)}(A) \right\}$
for each $A \subset I^{(n)}$ for $n$ large enough.
Then, for each $a \in {\cal N}^{(n)} = \{0,1,\ldots,n\}^2$,
we use the union bound to obtain a uniform bound
over all sets $A \subset I^{(n)}$ with $|A|_{\Sigma^{(n)}} = a$.
Finally, another use of the union bound over all $a \in {\cal N}^{(n)}$ gives us the result.

\subsection{Partial bound}

Let $n \ge 1$, $a \in {\cal N}^{(n)}$ and $A \subset I^{(n)}$ so that $|A|_{\Sigma^{(n)}} = a$.
Recall that $\mu^{(n)}(A) = \mathbb{E}[{\bf M}^{(n)}(A)]$ with
$$
{\bf M}^{(n)}(A) = \xi^{(n)} \sum_{s \in S^{(n)}} 1_{s \in \bigcup_{i \in A} {\bf S}_i^{(n)}}.
$$
The variables $1_{s \in \bigcup_{i \in A} {\bf S}_i^{(n)}}$ for $s \in S^{(n)}$ are Bernoulli distributed with parameter
$$
p_a^{(n)} = 1 - \prod_{k=1}^K \left( 1 - \frac{d_k^{(n)}}{m^{(n)}} \right)^{a_k}.
$$
Dubbashi et al. proved in Theorem $10$ of \cite{DPR96} that these random variables are negatively associated in the sense of Definition $3$ of \cite{DPR96}.
Their Theorem $14$ then showed that the Chernoff-Hoeffding bounds (see for instance \cite{M95,S11}), which hold for independent random variables,
can also be applied to these random variables.
Hence we have
\begin{align}
  \label{eq:smallchernoff}
  \mathbb{P} \left\{ {\bf M}^{(n)}(A) \le (1 - \epsilon) \mu^{(n)}(A) \right\} &\le e^{- \frac{\epsilon^2}2 m^{(n)} p_a^{(n)}}, \\
  \label{eq:bigchernoff}
  \mathbb{P} \left\{ {\bf M}^{(n)}(A) \le (1 - \epsilon) \mu^{(n)}(A) \right\} &\le e^{- m^{(n)} H\left[ (1 - \epsilon) p_a^{(n)} || p_a^{(n)} \right]},
\end{align}
where for any $p,q \in (0,1)$, $H[p||q]$ is the KL divergence between two Bernoulli random variables with parameters $p$ and $q$ respectively, given by
$$
H[p||q] = p \log\left( \frac{p}{q} \right) + (1-p) \log\left( \frac{1-p}{1-q} \right).
$$

We also use the following lemmas which will be proved later in Appendix \ref{app:techlemmas}:
\begin{lemma}
  \label{lem:smallminor}
  Let $0 < \delta < \frac12$.
  Consider a sequence $(g_n : n \ge 1)$ such that $g_n = o\left( d_1^{(n)} \right)$ and $g_n = o\left( d_2^{(n)} \right)$.
  For large enough $n$, we have
  $$
  p_a^{(n)} \ge \delta \frac{(a_1 + a_2) g_n}n,
  \quad \forall a = (a_1, a_2) \in \left\{ 0,1,\ldots,\left\lfloor \frac{n}{g_n} \right\rfloor \right\}^2.
  $$
\end{lemma}
\begin{lemma}
  \label{lem:bigminor}
  There exists a positive constant $\delta$ such that
  $$
  H\left[ (1-\epsilon) p_a^{(n)} || p_a^{(n)} \right] \ge -\delta + \epsilon \frac{a_1 d_1^{(n)} + a_2 d_2^{(n)}}{m^{(n)}},
  \quad \forall n \ge 1,
  \quad \forall a \in {\cal N}^{(n)}.
  $$
\end{lemma}

Consider the sequence $(g_n : n \ge 1)$ given by
$$
g_n = \left( \min\left( d_1^{(n)}, d_2^{(n)} \right) \log n \right)^{1/2},
\quad \forall n \ge 1.
$$
Observe that $g_n = \omega(\log n)$, $g_n = o\left( d_1^{(n)} \right)$ and $g_n = o\left( d_2^{(n)} \right)$.
Now let $n \ge 1$ and $a \in {\cal N}^{(n)}$. We distinguish two cases depending on the value of $a$.

\subsubsection{Case 1: $0 \le a_1 \le \frac{n}{g_n}$ and $0 \le a_2 \le \frac{n}{g_n}$}
\label{subapp:case1}

By Lemma \ref{lem:smallminor}, there is a positive constant $\delta_1$ such that, for large enough $n$,
$$
p_a^{(n)} \ge \delta_1 \frac{(a_1 + a_2) g_n}n.
$$
Using (\ref{eq:smallchernoff}), we deduce that
$$
\mathbb{P} \left\{ {\bf M}^{(n)}(A) \le (1 - \epsilon) \mu^{(n)}(A) \right\}
\le e^{- \frac{\epsilon^2}2 \delta_1 b (a_1 + a_2) g_n}
$$
for any $A \subset I^{(n)}$ such that $|A|_{\Sigma^{(n)}} = a$.
The union bound yields
\begin{align*}
  &\mathbb{P} \left\{ \exists A \subset I^{(n)} \text{ s.t. } |A|_{\Sigma^{(n)}} = a
  \text{ and } {\bf M}^{(n)}(A) \le (1 - \epsilon) \mu^{(n)}(A) \right\} \\
  &\le e^{- \frac{\epsilon^2}2 \delta_1 b (a_1 + a_2) g_n} \binom{n}{a_1} \binom{n}{a_2}, \\
  &\le e^{- \frac{\epsilon^2}2 \delta_1 b (a_1 + a_2) g_n} n^{a_1} n^{a_2}, \\
  &\le e^{a_1 g_n \left(- \frac{\epsilon^2}2 \delta_1 b + \frac{\log n}{g_n}\right)}
  e^{a_2 g_n \left(- \frac{\epsilon^2}2 \delta_1 b + \frac{\log n}{g_n}\right)}.
\end{align*}
Since $g_n = \omega(\log n)$, we obtain for large enough $n$
$$
\mathbb{P} \left\{ \exists A \subset I^{(n)} \text{ s.t. } |A|_{\Sigma^{(n)}} = a \text{ and } {\bf M}^{(n)}(A) \le (1 - \epsilon) \mu^{(n)}(A) \right\}
\le e^{- \delta_2 (a_1 + a_2) g_n}
$$
with $\delta_2 = \frac{\epsilon^2}4 \delta_1 b > 0$.

\subsubsection{Case $2$: $a_1 > \frac{n}{g_n}$ or $a_2 > \frac{n}{g_n}$}
\label{subapp:case2}

Combining Lemma \ref{lem:bigminor} with (\ref{eq:bigchernoff}),
we deduce that there is a positive constant $\delta_3$ such that
$$
\mathbb{P} \left\{ {\bf M}^{(n)}(A) \le (1 - \epsilon) \mu^{(n)}(A) \right\}
\le e^{\delta_3 m^{(n)} - \epsilon \left( a_1 d_1^{(n)} + a_2 d_2^{(n)} \right)}
$$
for any $A \subset I^{(n)}$ such that $|A|_{\Sigma^{(n)}} = a$.
Since $m^{(n)} = \lceil b n \rceil$ and $g_n = o\left( d_1^{(n)} \right)$,
we have $\delta_3 m^{(n)} \le \frac\epsilon2 \frac{n d_1^{(n)}}{g_n}$ when $n$ is large enough.
If $a_1 > \frac{n}{g_n}$, we also have that $\frac\epsilon2 \frac{n d_1^{(n)}}{g_n} \le \frac\epsilon2 a_1 d_1^{(n)}$ so that
$$
\delta_3 m^{(n)} - \epsilon \left( a_1 d_1^{(n)} + a_2 d_2^{(n)} \right)
\le - \frac\epsilon2 a_1 d_1^{(n)} - \epsilon a_2 d_2^{(n)} \le - \frac\epsilon2 \left( a_1 d_1^{(n)} + a_2 d_2^{(n)} \right)
$$
for large enough $n$.
The same argument holds by inverting $a_1$ and $a_2$ when $a_2 > \frac{n}{g_n}$,
so we conclude that there is a positive constant $\delta_4$ such that for large enough $n$, we have
$$
\mathbb{P} \left\{ {\bf M}^{(n)}(A) \le (1 - \epsilon) \mu^{(n)}(A) \right\}
\le e^{- \delta_4 \left( a_1 d_1^{(n)} + a_2 d_2^{(n)} \right)}
$$
for any $A \subset I^{(n)}$ such that $|A|_{\Sigma^{(n)}} = a$.
The union bound yields
\begin{align*}
  &\mathbb{P} \left\{ \exists A \subset I^{(n)} \text{ s.t. }
  |A|_{\Sigma^{(n)}} = a \text{ and } {\bf M}^{(n)}(A) \le (1 - \epsilon) \mu^{(n)}(A) \right\} \\
  &\le e^{- \delta_4 \left( a_1 d_1^{(n)} + a_2 d_2^{(n)} \right)} \binom{n}{a_1} \binom{n}{a_2}, \\
  &\le e^{- \delta_4 \left( a_1 d_1^{(n)} + a_2 d_2^{(n)} \right)} n^{a_1} n^{a_2}, \\
  &\le e^{a_1 d_1^{(n)} \left(- \delta_4 + \log n / d_1^{(n)} \right)} e^{a_2 d_2^{(n)} \left(- \delta_4 + \log n / d_2^{(n)} \right)}.
\end{align*}
Since $d_1^{(n)} = \omega(\log n)$ and $d_2^{(n)} = \omega(\log n)$, for large enough $n$, we have
\begin{align*}
  \mathbb{P} \left\{ \exists A \subset I^{(n)} \text{ s.t. }
  |A|_{\Sigma^{(n)}} = a \text{ and } {\bf M}^{(n)}(A) \le (1 - \epsilon) \mu^{(n)}(A) \right\}
  &\le e^{- \delta_5 a_1 d_1^{(n)}} e^{- \delta_5 a_2 d_2^{(n)}}
  = e^{- \delta_5 \left( a_1 d_1^{(n)} + a_2 d_2^{(n)} \right)}
\end{align*}
for some positive constant $\delta_5 < \delta_4$.

\subsection{Conclusion}

Combining cases \ref{subapp:case1} and \ref{subapp:case2},
we deduce that there exists a positive constant $\delta_6$ such that
$$
\mathbb{P} \left\{ \exists A \subset I^{(n)} \text{ s.t. }
|A|_{\Sigma^{(n)}} = a \text{ and } {\bf M}^{(n)}(A) \le (1 - \epsilon) \mu^{(n)}(A) \right\}
\le e^{- \delta_6 g_n},
\quad \forall a \in {\cal N}^{(n)}
$$
when $n$ is large enough.
Using the union bound again, we obtain
\begin{align*}
  \mathbb{P} \left\{ \exists A \subset I^{(n)} \text{ s.t. }
  {\bf M}^{(n)}(A) \le (1 - \epsilon) \mu^{(n)}(A) \right\}
  &\le n^2 e^{- \delta_6 g_n}
  = e^{2 \log n} e^{- \delta_6 g_n}
  = e^{- g_n \left( \delta_6 - \frac{2 \log n}{g_n} \right)}.
\end{align*}
Since $g_n = \omega(\log n)$, we conclude that for a constant $\delta_7 < \delta_6$, we have for large enough $n$
$$
\mathbb{P} \left\{ \exists A \subset I^{(n)} \text{ s.t. }
{\bf M}^{(n)}(A) \le (1 - \epsilon) \mu^{(n)}(A) \right\}
\le e^{- \delta_7 g_n}.
$$

Finally, when servers are in groups as in Assumption \ref{ass:servers},
we can break down ${\bf M}^{(n)}$ into a sum of random rank functions, one for each groups.
The result follows by showing the concentration in each group as above,
and then using the union bound again.

\section{Proof of the lemmas for Theorem \ref{theo:concentration}}
\label{app:techlemmas}

\setcounter{lemma}{2}
\begin{lemma}
  \label{lem:smallminor}
  Let $0 < \delta < \frac12$.
  Consider a sequence $(g_n : n \ge 1)$ such that $g_n = o\left( d_1^{(n)} \right)$ and $g_n = o\left( d_2^{(n)} \right)$.
  For large enough $n$, we have
  $$
  p_a^{(n)} \ge \delta \frac{(a_1 + a_2) g_n}n,
  \quad \forall a = (a_1, a_2) \in \left\{ 0,1,\ldots,\left\lfloor \frac{n}{g_n} \right\rfloor \right\}^2.
  $$
\end{lemma}

\begin{proof}
  Consider the sequence $(f_n : n \ge 1)$ of functions defined on $\mathbb{R}_+^2$ by
  $$
  f_n(t_1,t_2) = 1 - \left( 1 - \frac{d_1^{(n)}}{bn} \right)^{t_1} \left( 1 - \frac{d_2^{(n)}}{bn} \right)^{t_2}, \quad \forall (t_1,t_2) \in \mathbb{R}_+^2.
  $$
  We have
  \begin{align*}
    f_n\left( \frac{2n}{g_n}, 0 \right) = 1 - \left[ \left( 1 - \frac{d_1^{(n)}}{bn} \right)^\frac{n}{g_n} \right]^2 \xrightarrow[n \to \infty]{} 1
    \quad \text{and} \quad
    f_n\left( 0, \frac{2n}{g_n} \right) = 1 - \left[ \left( 1 - \frac{d_2^{(n)}}{bn} \right)^\frac{n}{g_n} \right]^2 \xrightarrow[n \to \infty]{} 1.
  \end{align*}
  Thus, there is $n_\delta \ge 1$ so that
  $f_n\left( \frac{2n}{g_n}, 0 \right) \ge 2 \delta$ and $f_n\left( 0, \frac{2n}{g_n} \right) \ge 2 \delta$ for all $n \ge n_\delta$.

  Then, for any $n \ge n_\delta$ and any $t_1, t_2 \le \frac{n}{g_n}$, we have
  \begin{align*}
    f_n(t_1, t_2)
    &= f_n \left( \frac{t_1}{t_1+t_2} (t_1 + t_2, 0) + \left( 1 - \frac{t_1}{t_1+t_2} \right) (0, t_1 + t_2) \right), \\
    &\ge \frac{t_1}{t_1+t_2} f_n(t_1+t_2,0) + \frac{t_2}{t_1+t_2} f_n(0,t_1+t_2), \\
    &\ge \frac{t_1}{t_1+t_2} \frac{t_1+t_2}{\frac{2n}{g_n}} f_n\left( \frac{2n}{g_n}, 0 \right)
    + \frac{t_2}{t_1+t_2} \frac{t_1+t_2}{\frac{2n}{g_n}} f_n\left( 0, \frac{2n}{g_n} \right), \\
    &\ge 2 \delta \frac{t_1 g_n}{2n} + 2 \delta \frac{t_2 g_n}{2n}, \\
    &= \delta \frac{(t_1 + t_2) g_n}{n},
  \end{align*}
  where the first two inequalities hold by concavity of $f_n$.
\end{proof}

\begin{lemma}
  \label{lem:bigminor}
  There exists a positive constant $\delta$ such that
  $$
  H\left[ (1-\epsilon) p_a^{(n)} || p_a^{(n)} \right] \ge -\delta + \epsilon \frac{a_1 d_1^{(n)} + a_2 d_2^{(n)}}{m^{(n)}},
  \quad \forall n \ge 1,
  \quad \forall a \in {\cal N}^{(n)}.
  $$
\end{lemma}

\begin{proof}
  By definition of $H$,
  \begin{align*}
    H\left[ (1 - \epsilon) p_a^{(n)} || p_a^{(n)} \right]
    {}={} &(1 - \epsilon) p_a^{(n)} \log(1 - \epsilon) \\
          &{}+{} \left( 1 - (1 - \epsilon) p_a^{(n)} \right) \log\left( 1 - (1 - \epsilon) p_a^{(n)} \right) \\
          &{}-{} \left( 1 - (1 - \epsilon) p_a^{(n)} \right) \log\left( 1 - p_a^{(n)} \right).
  \end{align*}
  The first and the second terms are greater
  than $(1 - \epsilon) \log(1 - \epsilon)$ and $\log(\epsilon)$ respectively.
  With $\delta = (1 - \epsilon) \log\left( \frac1{1-\epsilon} \right) + \log\left( \frac1\epsilon \right) > 0$, we obtain
  \begin{align*}
    H\left[ (1 - \epsilon) p_a^{(n)} || p_a^{(n)} \right]
    \ge - \delta - \left( 1 - (1 - \epsilon) p_a^{(n)} \right) \log\left( 1 - p_a^{(n)} \right)
    \ge - \delta - \epsilon \log\left( 1 - p_a^{(n)} \right).
  \end{align*}
  Finally, observe that
  \begin{align*}
    \log\left( 1 - p_a^{(n)} \right)
    = a_1 \log\left( 1 - \frac{d_1^{(n)}}{m^{(n)}} \right) + a_2 \log\left( 1 - \frac{d_2^{(n)}}{m^{(n)}} \right)
    \le - \frac{a_1 d_1^{(n)} + a_2 d_2^{(n)}}{m^{(n)}},
  \end{align*}
  where in the inequality we used the fact that $\log\left( 1 - \frac{d_k^{(n)}}{m^{(n)}} \right) \le - \frac{d_k^{(n)}}{m^{(n)}}$
  for $k = 1,2$.
  Hence, we obtain the expected result.
\end{proof}

\end{document}